\documentclass[12pt]{iopart} 

\usepackage{iopams, amsthm, amssymb, amsfonts}

\newtheorem{thm}{Theorem}
\newtheorem{lem}[thm]{Lemma}
\newtheorem{cor}[thm]{Corollary}

\pdfoutput=1
\usepackage{etex}
\reserveinserts{18}

\usepackage{cite}
\usepackage{latexsym}
\usepackage{rawfonts}
\usepackage{graphicx}
\usepackage[usenames,dvipsnames,svgnames]{xcolor}
\usepackage{bbm}
\usepackage{latexsym}
\usepackage{multirow}
\usepackage{rotating}
\usepackage{lscape}
\usepackage{graphicx} 
\usepackage{subfigure}
\usepackage{xcolor}
\usepackage{fancybox}
\usepackage{ifmtarg}

\input prepictex
\input pictex
\input postpictex

\usepackage{cmbright}
\usepackage[T1]{fontenc}
\usepackage{graphicx}

\newcommand{\Imi}{\mathrm{i}}

\newcommand{\Ref}[1]{(\ref{#1})}
\def\Bi#1#2{\L {{#1}\atop{#2}} \R}
\def\L{\left(}
\def\R{\right)}

\def\LV{\left|}
\def\RV{\right|}

\def\C#1{{\cal{#1}}}
\def\a#1{a_#1}

\def\q{\quad}

\def\sfrac#1#2{\hbox{\normalsize $\frac{#1}{#2}$}}
\def\Sfrac#1#2{\hbox{\large $\frac{#1}{#2}$}}

\usepackage{cmbright}
\usepackage{textcomp}
\usepackage[font=small,labelfont=sf]{caption}

\usepackage{subfig}

\definecolor{blue}{rgb}{0,0.18,0.39}
\definecolor{RoyalBlue}{rgb}{0,0.2,0.7}

\def\axes#1#2#3#4#5#6#7{
\setplotarea x from #7 to #5, y from #2 to #6
\setplotarea x from #1 to #5, y from #2 to #6
\axis left shiftedto x=#3 /
\axis bottom shiftedto y=#4 /
\put {\footnotesize$\bullet$} at #3 #4
}

\newcommand{\lmcn}{lima\c{c}on} 
\newcommand{\lmO}{\mathcal{L}_\mathrm{O}}
\newcommand{\lmI}{\mathcal{L}_\mathrm{I}}

\begin{document}
\title{Partition function zeros of adsorbing Dyck paths}
\author{
N.R. Beaton$^1$\footnote[1]{\texttt{nrbeaton@unimelb.edu.au}}
}

\address{$^1$School of Mathematics and Statistics, 
The University of Melbourne, VIC 3010, Australia\\}

\author{
E.J. Janse van Rensburg$^2$\footnote[2]{\texttt{rensburg@yorku.ca}}
}

\address{$^2$Department of Mathematics and Statistics, 
York University, Toronto, Ontario M3J~1P3, Canada\\}

\begin{abstract}
The zeros of the size-$n$ partition functions for a statistical mechanical model 
can be used to help understand the critical behaviour of the model as $n\to\infty$. 
Here we use weighted Dyck paths as a simple model of two-dimensional polymer 
adsorption, and study the behaviour of the partition function zeros, particularly 
in the thermodynamic limit. The exact solvability of the model allows for a precise 
calculation of the locus of the zeros and the way in which an edge-singularity 
on the positive real axis is formed.  We also show that in the limit
$n\to\infty$ the zeros converge on a lima\c{c}on in the complex plane.
\end{abstract}

\maketitle

\section{Introduction}
\label{section1}   

Lee-Yang \cite{LY52,YL52} and Fisher zeros \cite{F64,F65} have been examined as 
a mathematical mechanism whereby critical points appear in the thermodynamic 
limit in models in statistical mechanics. The Lee-Yang theorem (see, for
example, theorem 5.1.2 in reference \cite{Ruelle83}) shows that 
Lee-Yang zeros accumulate on the unit circle in the complex plane in a broad
class of lattice models in statistical mechanics, including the lattice gas
and Ising model.  Studies of Lee-Yang zeros include the the Potts model
\cite{KC98a,KC01}, lattice $\phi^3$ theories \cite{KF79} 
and the $n$-vector model \cite{C85}.

In a finite size model with complex temperature or interaction strength,
zeros in the partition function (the \emph{Fisher zeros}) are (complex) non-analyticities in the 
free energy.  In the thermodynamic limit these non-analyticities 
accumulate on a critical point on the positive real axis to form an 
\textit{edge-singularity}.  This mechanism is thought to be a mathematical 
mechanism whereby phase transitions arise in the thermodynamic
limit in models of critical phenomena.

Partition function zeros in the complex temperature plane have a
different distribution and do not generally accumulate on the unit circle.
The distribution is generally dependent on the model, and is useful in
modeling phase transitions \cite{F65}.  In an Ising model with
special boundary conditions the temperature plane zeros accumulate on
two circles corresponding to the ferromagnetic and anti-ferromagnetic
phases, a  pattern which is not seen in the $q$-Potts model when
$q>2$ \cite{KC98a,KC01}.

In reference \cite{JvR17} the properties of partition and generating function 
zeros in a self-avoiding walk model of polymer adsorption were examined.
While some generating function zeros were shown to accumulate on a 
circle in the complex plane, the partition function zeros appear to approach
a limiting distribution which forms an edge-singularity at the critical
adsorption point.  In this paper our aim is to examine partition function
zeros of adsorbing Dyck paths as a directed version of the adsorbing
self-avoiding walk model in reference \cite{JvR17}.

Adsorbing walks are models of the adsorption transition of a dilute polymer 
in a good solvent.  The adsorption transition in a polymer is characterised by a 
conformational rearrangement of monomers, and by singular behaviour 
of thermodynamic quantities at the adsorption critical point 
(see for example references \cite{DL93,deG79,HLK04,HMK04}).  The adsorbing 
self-avoiding walk is a standard model for linear polymer adsorption 
\cite{HTW82,FPF89,VW98A} and has been studied numerically in, for example, 
references \cite{JvRR04,JvR16}.  Exactly solvable models of directed 
lattice paths were introduced in references \cite{PFF88,W98} as models of
linear polymer adsorption; see references \cite{BEO98,JvR10A} as well.   
These directed models may be exactly solvable, and considerably more information can 
be obtained by analysing them \cite{JvR03}.  This may give some information 
on the adsorption critical point as the thermodynamic limit is taken 
in the more general case.

\begin{figure}[t]
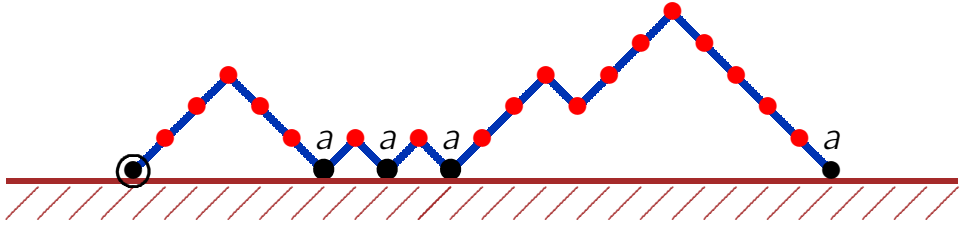

\beginpicture
\setcoordinatesystem units <1.2pt,1.2pt>
\setplotarea x from -40 to 150, y from -50 to 50
\setlinear

\color{Brown}
\multiput
{\plot 0 0 10 10 /}
at -20 -60 -10 -60 0 -60  10 -60 20 -60 30 -60 
   40 -60 50 -60    60 -60 70 -60 80 -60 90 -60 
   100 -60 110 -60 120 -60 130 -60 140 -60  
   150 -60 160 -60 170 -60 180 -60 190 -60
   200 -60 210 -60 220 -60 230 -60 240 -60 
   250 -60 110 -60 260 -60 270 -60 /

\setplotsymbol ({\huge.})
\color{RoyalBlue}
\plot 20 -50 30 -40 40 -30 50 -20 60 -30 70 -40 80 -50 
90 -40 100 -50 110 -40 120 -50 130 -40 140 -30 150 -20
160 -30 170 -20 180 -10 190 0 200 -10 210 -20 220 -30 
230 -40 240 -50 /

\color{red}
\multiput {\Large $\bullet$} at 
30 -40 40 -30 50 -20 60 -30 70 -40 80 -50 
90 -40 100 -50 110 -40 120 -50 130 -40 140 -30 150 -20
160 -30 170 -20 180 -10 190 0 200 -10 210 -20 220 -30 
230 -40 /
\color{black}
\multiput {\LARGE $\bullet$} at 80 -50 100 -50 120 -50 /
\multiput {\large $a$} at  80 -40 100 -40 120 -40 240 -40 /

\color{Brown}
\linethickness=2pt
\putrule from -20 -53 to 280 -53

\color{black}
\setplotsymbol ({.})
\circulararc 360 degrees from 20 -55 center at 20 -50
\multiput {\Large $\bullet$} at 20 -50 / 
\multiput {\Large $\bullet$} at 240 -50 /
\normalcolor

\endpicture
\caption{An adsorbing Dyck path of length $22$.  The path gives North-East and South-East steps 
in the positive square lattice (the $x$-axis is a hard wall) and is weighted by the 
number of returns (visits) to the wall (or adsorbing line), in this case $a^4$. The path is
conditioned to end in the adsorbing line.}
\label{figure1}
\end{figure}

A \emph{Dyck path} of length $2n$ is a walk on the square lattice 
$\mathbb Z^2$, starting at $(0,0)$ and ending at $(2n,0)$, taking steps 
$(1,1)$ and $(1,-1)$, and always remaining on or above the line $y=0$ 
(see Figure \ref{figure1}).  Let $d_{2n}(v)$ be the number of 
Dyck paths of length $2n$ which contain $v+1$ vertices in the line $y=0$ (these are 
called \emph{visits}). We associate a weight $a$ with each visit (excluding the first vertex) 
to obtain a partition function $D_{2n}(a)$, given by
\begin{equation}
D_{2n} (a) = \sum_v d_{2n}(v)a^v = \sum_{\ell=0}^n \frac{2\ell+1}{n+\ell+1}\, \Bi{2n}{n+\ell}\, (a-1)^\ell \, ,
\label{eqn:exact_pf}
\end{equation}
see for example equation (5.32) in reference \cite{JvR15}.  

The partition function zeros $\a{\ell}$ of adsorbing Dyck paths of length $2n$ 
are the complex solutions of $D_{2n}(a)=0$ in the $a$-plane.  Since
$D_{2n}(a)$ is a polynomial of degree $n$ in $a$ with non-negative coefficients, 
and the coefficient of $a^n$ is equal to $1$, the zeros $\a{\ell}$ occur
as negative real numbers or as conjugate pairs, and $D_{2n}(a)$ factors as
\begin{equation}
D_{2n}(a) = \prod_{\ell=1}^n(a-\a{\ell}) .
\end{equation}
Since every non-empty Dyck path ends with a visit, $a=0$ is a root of 
$D_{2n}(a)$ for all $n\geq 1$. We will henceforth refer to $a=0$ as 
the \emph{trivial zero}.

In Section~\ref{section2} we examine the generating function of adsorbing Dyck paths in order to determine the asymptotic behaviour of $D_{2n}(a)$ for real and complex $a$, and find a ``phase diagram'' of sorts in the complex $a$-plane. In Section~\ref{section3} we turn to the zeros of $D_{2n}(a)$, and show that they collect and are dense on the curve that separates the two phases in the $a$-plane. In Section~\ref{section4} we use numerical techniques to estimate the locations of the zeros of $D_{2n}(a)$ for finite $n$. Finally in Section~\ref{sec:leadingzero} we compute the exact asymptotics of the ``leading'' zero (the one with smallest positive argument), which approaches the critical point $a_c =2$ as $n\to\infty$, this also being the location of a edge-singularity of the model.

\section{The generating function of adsorbing Dyck paths and the complex $a$-plane}
\label{section2}

For $a$ real and positive, the limiting free energy of adsorbing Dyck paths is given by
\begin{equation}
\C{D}(a) = \lim_{n\to\infty} \sfrac{1}{2n} \log D_{2n}(a) = 
\cases{
\log 2 & \textsf{if $a\leq 2$}; \\
\log a - \sfrac{1}{2} \log (a-1) & \textsf{if $a\geq 2$}
}
\label{eqn3}  
\end{equation}
$\C{D}(a)$ is singular at the point $a=2$, corresponding to the adsorption phase transition.  It is not immediately clear on how to generalise $\C{D}(a)$ as a
function of complex $a$.

The generating function of adsorbing Dyck paths is given by
\begin{equation}
D(t,a) = \sum_{n=0}^\infty D_{2n}(a)\, t^n 
= \frac{2}{2-a(1-\sqrt{1-4t})} .
\label{eqn5}  
\end{equation}
It follows that $D(t,a)$ has one or two singularities in the $t$-plane, 
depending on the value of $a$: a square root singularity at $t=\frac14$ for all $a\neq0$ or 2,  and a simple pole at $t=\frac{a-1}{a^2}$ when 
$|a-1|\geq1$ and $a\neq 0$ or 2. When $a=2$ these singularities 
coalesce and form a pole of order $\frac12$. The \emph{dominant} 
singularity (that is, the one closest to the origin) is then
\begin{equation}\
t_c(a) = \cases{
\sfrac14, & \textsf{if $|\frac{a-1}{a^2}| \geq \frac14$ or $|a-1| < 1$}; \\
\sfrac{a-1}{a^2}, & \textsf{if $|\frac{a-1}{a^2}| < \frac14$ and $|a-1| \geq 1$.}
}
\label{tc.Cplane} 
\end{equation}

Note that when $a$ is real and positive, $\mathcal D(a) = -\log t_c(a)$,
 as given by~\Ref{eqn3} and~\Ref{tc.Cplane}. This then answers the 
question as to the generalisation of $\mathcal D(a)$ to complex $a$: 
we should use $-\log t_c(a)$, which is exactly the analytic continuation 
of the two branches of~\Ref{eqn3} which gives a continuous function in $\mathbb C$.

The curve in $\mathbb{C}$ which delineates the two regions given 
in~\Ref{tc.Cplane} is a branch of the curve
\begin{equation}
\LV \frac{a-1}{a^2} \RV = \frac{1}{4}.
\label{eqn4}   
\end{equation}
The curve defined by~\Ref{eqn4} is a \emph{lima\c{c}on}; it can be 
parametrised by $a=x+y\Imi$, where
\begin{equation}
x = 2 + \left(2\sqrt{2}-4\cos\phi\right)\cos\phi \quad\textsf{and}\quad y = \left(2\sqrt{2}-4\cos\phi\right)\sin\phi
\label{eqn5A}   
\end{equation}
for $\phi\in[0,2\pi)$.
The lima\c{c}on has two ``lobes'': taking only $\phi\in[\frac{\pi}{4},\frac{7\pi}{4})$ 
gives the outer lobe, while $\phi\in[0,\frac{\pi}{4})\cup(\frac{7\pi}{4},2\pi)$ gives the inner lobe (see Figure \ref{figure2}).  The inner lobe lies entirely within the region $|a-1|\leq 1$, where the simple pole does not exist, so that the dominant singularity is still at $t=\frac14$ there.  As a result, it is precisely the outer lobe of the lima\c{c}on which separates the two regions defined by~\Ref{tc.Cplane}; as $a$ crosses from one side of this curve to the other, $t_c(a)$ switches from one value to the other, and hence so too does $\mathcal D(a)$. We will use $\lmO$ to denote the outer lobe of the \lmcn{} and $\lmI$ to denote the inner lobe.

\begin{figure}[t]
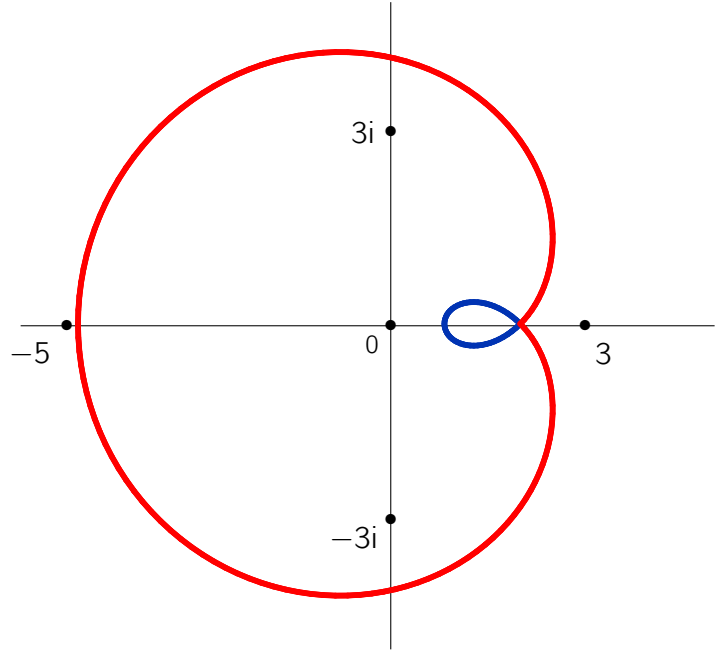

\begin{center}
\beginpicture
\color{black}
\setcoordinatesystem units <3.5pt,3.5pt>
\axes{-40}{-35}{0}{0}{35}{35}{-70}
\put {\footnotesize$0$} at -2 -2
\put {\footnotesize$\bullet$} at -35 0
\put {$-5$} at -39 -3
\put {\footnotesize$\bullet$} at 21 0
\put {$3$} at 23 -3
\put {\footnotesize$\bullet$} at 0 21
\put {$3\Imi$} at -3 21
\put {\footnotesize$\bullet$} at 0 -21
\put {$-3\Imi$} at -4 -23

\color{RoyalBlue}
\setplotsymbol ({\LARGE$\cdot$})
\plot 14.000 0.000  13.554 0.419  13.097 0.796  12.632 1.132  12.163 1.425  11.691 1.677  11.222 1.888  10.756 2.059  10.297 2.190  9.848 2.282  9.412 2.338  8.991 2.357  8.587 2.342  8.203 2.295  7.841 2.217  7.504 2.111  7.192 1.978  6.909 1.821  6.655 1.642  6.431 1.444  6.240 1.229  6.083 1.000  5.959 0.760  5.870 0.511  5.817 0.257  5.799 0.000  5.817 -0.257  5.870 -0.511  5.959 -0.760  6.083 -1.000  6.240 -1.229  6.431 -1.444  6.655 -1.642  6.909 -1.821  7.192 -1.978  7.504 -2.111  7.841 -2.217  8.203 -2.295  8.587 -2.342  8.991 -2.357  9.412 -2.338  9.848 -2.282  10.297 -2.190  10.756 -2.059  11.222 -1.888  11.691 -1.677  12.163 -1.425  12.632 -1.132  13.097 -0.796  13.554 -0.419  14.000 0.000 /

\color{red}
\plot 14.000 0.000  14.432 0.460  14.848 0.962  15.244 1.503  15.617 2.084  15.964 2.703  16.282 3.358  16.570 4.049  16.823 4.774  17.040 5.529  17.218 6.315  17.354 7.127  17.447 7.965  17.494 8.825  17.494 9.705  17.444 10.601  17.344 11.511  17.192 12.432  16.987 13.361  16.727 14.295  16.412 15.229  16.042 16.161  15.615 17.088  15.133 18.005  14.594 18.910  14.000 19.799  13.350 20.668  12.646 21.515  11.889 22.334  11.079 23.125  10.218 23.881  9.307 24.602  8.349 25.283  7.344 25.922  6.297 26.514  5.208 27.059  4.081 27.552  2.917 27.992  1.721 28.376  0.494 28.702  -0.760 28.967  -2.037 29.171  -3.334 29.310  -4.648 29.384  -5.975 29.392  -7.311 29.333  -8.653 29.204  -9.997 29.007  -11.339 28.741  -12.674 28.405  -14.000 28.000  -15.312 27.526  -16.606 26.983  -17.879 26.372  -19.126 25.695  -20.344 24.952  -21.529 24.146  -22.678 23.276  -23.786 22.347  -24.852 21.359  -25.870 20.315  -26.839 19.217  -27.754 18.069  -28.614 16.872  -29.416 15.631  -30.156 14.347  -30.833 13.025  -31.445 11.668  -31.990 10.280  -32.465 8.864  -32.870 7.423  -33.203 5.963  -33.463 4.487  -33.650 2.998  -33.762 1.501  -33.799 -0.000  -33.762 -1.501  -33.650 -2.998  -33.463 -4.487  -33.203 -5.963  -32.870 -7.423  -32.465 -8.864  -31.990 -10.280  -31.445 -11.668  -30.833 -13.025  -30.156 -14.347  -29.416 -15.631  -28.614 -16.872  -27.754 -18.069  -26.839 -19.217  -25.870 -20.315  -24.852 -21.359  -23.786 -22.347  -22.678 -23.276  -21.529 -24.146  -20.344 -24.952  -19.126 -25.695  -17.879 -26.372  -16.606 -26.983  -15.312 -27.526  -14.000 -28.000  -12.674 -28.405  -11.339 -28.741  -9.997 -29.007  -8.653 -29.204  -7.311 -29.333  -5.975 -29.392  -4.648 -29.384  -3.334 -29.310  -2.037 -29.171  -0.760 -28.967  0.494 -28.702  1.721 -28.376  2.917 -27.992  4.081 -27.552  5.208 -27.059  6.297 -26.514  7.344 -25.922  8.349 -25.283  9.307 -24.602  10.218 -23.881  11.079 -23.125  11.889 -22.334  12.646 -21.515  13.350 -20.668  14.000 -19.799  14.594 -18.910  15.133 -18.005  15.615 -17.088  16.042 -16.161  16.412 -15.229  16.727 -14.295  16.987 -13.361  17.192 -12.432  17.344 -11.511  17.444 -10.601  17.494 -9.705  17.494 -8.825  17.447 -7.965  17.354 -7.127  17.218 -6.315  17.040 -5.529  16.823 -4.774  16.570 -4.049  16.282 -3.358  15.964 -2.703  15.617 -2.084  15.244 -1.503  14.848 -0.962  14.432 -0.460  14.000 0.000   /

\color{black}
\normalcolor

\endpicture
\end{center}
\caption{The lima\c{c}on defined by equation \Ref{eqn4} has two lobes.  The inner lobe
lies inside the region $|a-1|\leq 1$ and plays no role in the zeros of $D_{2n}(a)$.  
The outer lobe separates the two regions defined by equation \Ref{tc.Cplane}: the dominant singularity is at $t=\frac14$ when $a$ is inside, and $t=\frac{a-1}{a^2}$ when $a$ is outside.}
\label{figure2}
\end{figure}

\section{Complex zeros are dense on the \lmcn{} in the limit $n\to\infty$}
\label{section3}

In this section we will see that the curve $\lmO$ not only divides the complex $a$-plane into two regions according to the asymptotic behaviour of $t_c(a)$ and $\mathcal{D}(a)$; it is also the case that all non-trivial roots of $D_{2n}(a)$ approach $\lmO$ as $n\to\infty$, and that the roots become dense on $\lmO$ in the limit. In this section we will assume that $a\neq0$ or 2.

The two singularities in $D(t,a)$ (in equation \Ref{eqn5})
contribute to the exponential
growth of $D_{2n}(a)$.  Expanding $D(t,a)$ around the point $t=\frac14$ gives
\begin{equation}
D(t,a) = \frac{2}{2-a} \sum_{j=0}^\infty \L \frac{a}{a-2} \R^j\, (1-4t)^{j/2}.
\label{eqn7}  
\end{equation}
This series becomes singular if $a\to2$. To avoid this,  bound $a$ away from 2
by fixing $\delta>0$ and defining
\begin{equation}
S_\delta = \{ z\in {\mathbb C} \, \hbox{\Large$|$}\, \hbox{$|z-2|\geq\delta$} \}.
\end{equation}
In what follows, assume that $a\in S_\delta$. 
Later in Section~\ref{sec:leadingzero} we will address the case $a\to2$.

By the Cauchy integral theorem, the terms in equation \Ref{eqn7}
contribute to the exponential growth of $D_{2n}(a)$.  Computing
the coefficients of $t^n$ in equation \Ref{eqn7} shows that $D_{2n}(a)$ 
can be cast in the form
\begin{equation}
r_{n}(a) = \frac{a}{(a-2)^2} \cdot \frac{4^n}{\sqrt{\pi n^3}} \sum_{\ell=0}^\infty \frac{g_\ell(a)}{n^{\ell}} .
\label{sqrt.asymp}  
\end{equation}
This is done by expanding the factor 
\begin{equation}
(1-4t)^{j/2} = \sum_{\ell=0}^\infty \L -\Sfrac{j}{2} \R_\ell \frac{(4t)^\ell}{\ell !}  
\end{equation}
where $(a)_n$ is the Pochhammer function, and then determining
the functions $g_\ell (a)$ term-by-term.  Notice that
$D_{2n}(a) = r_n(a)$ if $|a-1|<1$.  The functions 
$g_\ell(a)$ are rational functions of $a$ of the form
\begin{equation}\label{eqn:g_ell_form}
g_\ell(a) = \frac{q_\ell(a)}{(a-2)^\ell}
\end{equation}
and $q_\ell(a)$ are computable polynomials of degree at most $\ell -1$ in $a$. 
The sum in equation \Ref{sqrt.asymp} 
can be rewritten to give
\begin{equation}
r_{n}(a) = \frac{a}{(a-2)^2} \cdot \frac{4^n}{\sqrt{\pi n^3}} \left(\sum_{\ell=0}^L \frac{g_\ell(a)}{n^{\ell}} + \frac{\beta_{L+1,n}(a)}{n^{L+1}}\right)
\label{sqrt.asymp.trunc}
\end{equation}
for any $L\geq -1$ (if $L=-1$ then the sum is empty), where 
(as long as $a\in S_\delta$) $\beta_{L+1,n}(a) \to g_{L+1}(a)$ as $n\to\infty$. 
That is, the sum in equation \Ref{sqrt.asymp} can be truncated at any 
$\ell\geq0$ and what remains will be a valid asymptotic expansion 
for the contribution of equation \Ref{eqn7}.

Expanding $D(t,a)$ in equation \Ref{eqn5} near the simple pole gives
\begin{equation}
D(t,a) = \frac{a-2}{a-1} \, (1-At)^{-1}
  + \sum_{k=0}^\infty \frac{(-1)^k (a-1)^k}{(a-2)^{2k+1}} \, C_k \, (1-At)^k
\label{eqn8}  
\end{equation}
where $C_k$ are Catalan numbers and $A=\Sfrac{a^2}{a-1}$.
Determining the coefficients $p_n(a)$ of $t^n$ gives, when $|a-1| \geq 1$, 
\begin{equation}
p_n(a) = \frac{a-2}{a-1} \L \frac{a^2}{a-1} \R^{\!\!n}.
\label{pole.asymp}  
\end{equation}
Notice that there are no additional terms; the contribution of the series in equation
\Ref{eqn8} is already given by equation \Ref{sqrt.asymp.trunc}.
See also, for example, equation (5.17) in reference \cite{JvR15}.

\begin{lem}\label{lem:no_roots_near_1}
Let $\epsilon>0$ and define 
$B_\epsilon = \{ z\in {\mathbb C} \, 
\hbox{\Large$|$}\, \hbox{$|z-1|\leq 1-\epsilon$} \}$. 
Then for all $n$ sufficiently large, there are no zeros of $D_{2n}(a)$ inside $B_\epsilon$.
\end{lem}
\begin{proof}
If $a\in B_\epsilon$ then the only singularity of $D(t,a)$ is the square 
root singularity at $t=\frac14$. Thus $D_{2n}(a) = r_n(a)$ as 
per~(\ref{sqrt.asymp}) or (\ref{sqrt.asymp.trunc}). Since $g_0(a) = 1$, 
truncate the sum to obtain the asymptotic form
\begin{equation}
r_n(a) \sim r^0_{n}(a) = \frac{a}{(a-2)^2}\cdot\frac{4^n}{\sqrt{\pi n^3}}.
\end{equation}
For $a\in B_\epsilon$, $D_{2n}(a)$ can be made arbitrarily 
close to $r^0_{n}(a)$ by taking $n$ sufficiently large 
(since $r_n^0(a)/r_n(a) \to 1$ as $n\to\infty$). 
But clearly $r_n(a) \neq 0$ for $a\in B_\epsilon$, so $D_{2n}(a)$ 
cannot have a zero there.
\end{proof}

Outside of $B_\epsilon$ things are more interesting.

\begin{lem}\label{lem:roots_to_lmcn}
For each $n\geq1$, let $\alpha_n$ be any sequence of non-trivial zeros
of $D_{2n}(a)$ in $S_\delta$. Then
\begin{equation}
\LV \frac{\alpha_n^2}{\alpha_n-1} \RV \to 4 \quad \mathit{as}\quad n\to\infty.
\label{alpha.to.4}
\end{equation}
\end{lem}
\begin{proof}
If $|a-1|<1$ then $D_{2n}(a) = r_n(a)$, and by 
Lemma~\ref{lem:no_roots_near_1} 
there are no possible zeros for large $n$.   Thus, assume that 
$|a-1|\geq1$ so that $D_{2n}(a) = r_n(a) + p_n(a)$.  If $\alpha_n$ is a zero of
$D_{2n}(a)$, then we must have $|r_n(\alpha_n)| = |p_n(\alpha_n)|$.
In this case $|r_n(\alpha_n)|$ grows at the exponential rate
$4^n$ and $|p_n(\alpha_n)|$ grows at the exponential rate
$\LV \frac{\alpha_n^2}{\alpha_n-1} \RV^n$.  If $n$ is large
then for $\alpha_n$ to be a zero it must (approximately) 
balance the two exponential growth rates for finite values of $n$.   
As $n\to\infty$ then the two growth rates must become equal. 
This can only happen if $\alpha_n$ satisfies equation \Ref{alpha.to.4}.

Note that at first glance it may appear that $\alpha_n\to0$ is also 
a valid possibility, but a more careful analysis shows that in that 
case $p_n(\alpha_n)$ would approach $0$ exponentially faster 
than $r_n(\alpha_n)$. 
\end{proof}

Lemmas~\ref{lem:no_roots_near_1} and~\ref{lem:roots_to_lmcn} immediately lead to the following.
\begin{cor}\label{cor:outer_lobe_roots}
As $n\to\infty$, all non-trivial zeros of $D_{2n}(a)$ approach the curve $\lmO$.
\end{cor}

Corollary~\ref{cor:outer_lobe_roots} only establishes that as $n$ gets large, all zeros of $D_{2n}(a)$ will be \emph{somewhere} near $\lmO$. We also wish to show the converse of this result -- that every point on $\lmO$ is an accumulation point of zeros of $D_{2n}(a)$, so that the zeros become dense on $\lmO$ as $n\to\infty$. We will prove the following.

\begin{thm}\label{thm:dense}
Let $\epsilon,\delta>0$. Then for every every $n$ sufficiently large and for every point $a^*\in \lmO\cap S_\delta$, there is a zero of $D_{2n}(a)$ in the region $|a-a^*| \leq \epsilon$.
\label{theorem4}   
\end{thm}

To show this, we will estimate the locations of the zeros by using an approach similar to that used in reference \cite{BM10}, where it was shown that poles of a certain generating function become dense on parts of the unit circle. 

Define the approximation
\begin{equation}
D^0_{2n}(a) = p_n(a) + r^0_{n}(a) = \frac{a-2}{a-1} \L \frac{a^2}{a-1} \R^n
+ \frac{a}{(a-2)^2} \frac{4^n}{\sqrt{\pi n^3}}. 
\label{eqn16}  
\end{equation}
For $|a-1| \geq 1$, $D_{2n}^0(a)$ is an asymptotic approximation 
for $D_{2n}(a)$; that is, $D_{2n}^0(a)/D_{2n}(a) \to 1$ as $n\to\infty$. 
It follows that, in this region, the zeros of $D_{2n}^0(a)$ give 
asymptotic approximations to the zeros of $D_{2n}(a)$. Note that 
$D_{2n}^0(a)$ is not a good approximation in the region $|a-1|<1$; 
as we will see later, $D_{2n}^0(a)$ has zeros inside this region 
which collect on the inner lobe $\lmI$, while $D_{2n}(a)$ does not.

The form~\Ref{eqn4} of the lima\c{c}on suggests that a change 
of variables will be convenient. Define $A=\sfrac{a^2}{a-1}$ and 
solve for $a$ in terms of $A$:
\begin{equation}
a = a^+ \equiv a^+(A) = \sfrac{1}{2} (A + \sqrt{A}\; \sqrt{A-4}) .
\label{eqn12}   
\end{equation}
When $|A|=4$, the solution $a^+$ is on $\lmO$. (The other solution 
to the quadratic gives a point on $\lmI$.) Setting
\[ f(a) = \frac{a-2}{a-1}\q\hbox{and}\q g(a) = \frac{a}{(a-2)^2},\]
we then have
\begin{equation}
D_{2n}^0(a^+) = f(a^+) A^n
+ g(a^+) \frac{4^n}{\sqrt{\pi n^3}}. \, 
\label{eqn16alt}  
\end{equation}

\begin{proof}[Proof of Theorem~\ref{thm:dense}]
Assume that the complex $A$-plane has a branch cut along the positive 
real axis. Fix integers $p$ and $q$ so that $\theta=\sfrac{p}{q}\pi\in(0,2\pi)$ 
and $a^+(4e^{\Imi\theta}) \in S_\delta$. Set $n=2kq$. We will show that $4e^{\Imi\theta}$ is an accumulation point for roots of $D^0_{2n}(a^+)$ 
in the $A$-plane. This in turn means that $a^+(4e^{\Imi\theta})$ is an 
accumulation point for roots of $D^0_{2n}(a)$, and since $D^0_{2n}(a)$ 
is an asymptotic approximation for $D_{2n}(a)$, the result will follow.

Set $A_n = 4e^{\Imi\theta}(1+s_n)$ for an as-yet (small) unknown $s_n$ 
and write $a^+_n = a^+(A_n)$. Then  substitute $A=A_n$ 
in~\Ref{eqn16alt}. To leading order, this gives
\begin{equation}\label{eqn:sub_An_approx}
\hspace{-2cm}
D_{2n}^0(a^+_n) =
F(4\,e^{\Imi \theta})\L 4^n (1+s_n)^n \R \L 1+O(s_n) \R
 + G(4\,e^{\Imi \theta}) \frac{4^n}{\sqrt{\pi n^3}} \L 1+O(s_n) \R
\end{equation}
where $F(A) = f(a_+)$ and $G(A) = g(a_+)$. (This uses the analyticity 
of $f$ and $g$ away from $a=2$.)

Notice that $F(4\,e^{\Imi \theta})$ is nonzero except when $\theta=0$ 
(this is excluded), and that $G(4\,e^{\Imi \theta})$ is nonzero. Equation 
(\ref{eqn:sub_An_approx}) thus shows that there is a zero of 
(\ref{eqn16alt}) close to $A=4e^{\Imi\theta}$ for our chosen $\theta$. 
This solution is well approximated by ignoring the $O(s_n)$ terms, 
approximating $1+s_n = e^{s_n+O(s_n^2)}$, 
and solving for $s_n$ in (\ref{eqn:sub_An_approx}). This gives
the following approximation for $s_n$:
\begin{equation}
s_n \sim 
\sfrac{1}{n} \log\L \frac{-1}{\sqrt{\pi n^3}} \,
\frac{G(4\, e^{\Imi \theta})}{F (4\, e^{\Imi \theta})} \R
= \sfrac{1}{n} \log\L \frac{-1}{\sqrt{\pi n^3}} \,
\frac{e^{\Imi \theta/2}}{4\,( e^{\Imi \theta}-1)^{3/2} } \R .
\end{equation}
This shows that $s_n \to 0$ as $n\to\infty$, and hence $a^+_n\to 4e^{\Imi\theta}$. Since $\theta$ was an arbitrary rational multiple of $\pi$ and the circle $|A|=4$ maps continuously to $\lmO$ in the $a$-plane, it follows that for
$a^* \in \lmO \cap S_\delta$ there is a choice of $\theta$ and an 
$n$ sufficiently large so that there is a root of $D^0_{2n}(a)$ within 
the region $|a-a^*|\leq \epsilon$. Then since $D^0_{2n}(a) \to D_{2n}(a)$,
the result of the theorem follows.
\end{proof}

Note that more precise approximations for $s_n$ can be obtained by 
replacing $r_n^0(a)$ in equation \Ref{eqn16} with higher-order 
truncations of $r_n(a)$, for example
\begin{equation}
r_n^1(a) = \frac{a}{(a-2)^2}\cdot\frac{4^n}{\sqrt{\pi n^3}}\left(1 - \frac{3a^2}{2(a-2)^2n}\right).
\end{equation}
This changes the form of $g(a)$ and $G(A)$, but since all the $g_\ell(a)$ have a form given by (\ref{eqn:g_ell_form}), the analyticity away from $a=2$ still holds and all of the above steps will still be valid.

\section{Numerical approximation of the zeros}
\label{section4}

The method used to prove Theorem~\ref{thm:dense} gives an approximation 
for the zero of $D_{2n}(a)$ closest to any point in $\lmO \cap S_\delta$. 
However, this does not give much insight as to how the zeros ``move'' as $n$ changes. 
For example, the non-trivial zeros of $D_{2n}(a)$ can be ordered according to their 
argument, and then for some given $k$ (either constant or depending on $n$) the 
behaviour of the $k$-th zero can be investigated. In this section and the next, we 
will pursue this question, using two different methods -- one which works well away 
from $a=2$, and another which specifically probes the $a\to2$ regime.

We continue to assume $|a-1|\geq 1$, so that $D_{2n}(a) = p_n(a) + r_n(a)$.
If $a$ is a zero of $D_{2n}(a)$ then we can write, using equations 
\Ref{sqrt.asymp} and \Ref{pole.asymp}
\begin{equation}
 \L \frac{a^2}{a-1} \R^n = -\L \frac{4^n}{\sqrt{\pi n^3}} \R\!
\L \frac{a(a-1)}{(a-2)^3} + \frac{(a-1)\beta_{1,n}(a)}{n(a-2)} \R .
\label{eqn23}
\end{equation}
This may be put in the form 
\begin{equation}
a^2 = 4(a-1) e^{(2k+1)\pi\Imi/n} \L \frac{1}{\sqrt{\pi n^3}} \R^{\! 1/n}\!
\L \frac{a(a-1)}{(a-2)^3} + \frac{(a-1)\beta_{1,n}(a) }{n(a-2)} \R^{\! 1/n}
\label{eqn24}  
\end{equation}
for $k=0,1,2,\ldots,n-1$. 
The function $\beta_{1,n}(a)$, to leading order, is
given by $ \beta_{1,n}(a) = \Sfrac{3a^3}{2(a-2)^4} + O\L \Sfrac{1}{n} \R$.
For values of $a$ away from the critical value $a=2$, $|\beta_{1,n}(a)|$
is quite small (less than $1$), but at the leading zeros the modulus of 
$\beta_{1,n}(a)$ can be larger.  Nevertheless, since $\beta_{1,n}(a)$ is divided
by $n$ in equation \Ref{eqn24}, it is expected that its contribution will decrease
quickly with increasing $n$, and numerical experimentation gives results
consistent with this.  In figure \ref{figure3} we show the results for
$\beta_{1,n}(a)$ set equal to $4$.

We will now attempt to compute approximate solutions to (\ref{eqn24}) 
by replacing $\beta_{1,n}(a)$ with a \emph{constant} $\beta$. Note 
that if $\beta=0$ then we are finding roots of $D^0_{2n}(a)$ as per
equation (\ref{eqn16}). In Figure \ref{figure3} the solutions for $n\in\{16,32,48\}$ are shown
with $\beta=4$ (open circles).  The exact solutions are shown 
by the bullets for these values of $n$, and the lima\c{c}on
in Figure \ref{figure2} is the closed curve.  

\begin{figure}[t]
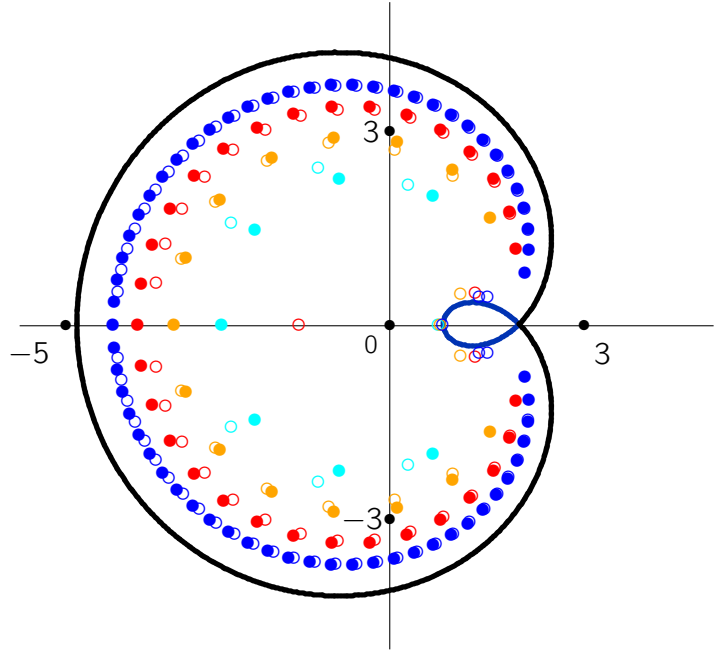

\begin{center}
\beginpicture
\color{black}
\setcoordinatesystem units <3.5pt,3.5pt>
\axes{-40}{-35}{0}{0}{35}{35}{-70}
\put {\footnotesize$0$} at -2 -2
\put {\footnotesize$\bullet$} at -35 0
\put {$-5$} at -39 -3
\put {\footnotesize$\bullet$} at 21 0
\put {$3$} at 23 -3
\put {\footnotesize$\bullet$} at 0 21
\put {$3$} at -2 21
\put {\footnotesize$\bullet$} at 0 -21
\put {$-3$} at -3 -21

\color{RoyalBlue}
\setplotsymbol ({\Large$\cdot$})
\plot 14.000 0.000  13.554 0.419  13.097 0.796  12.632 1.132  12.163 1.425  11.691 1.677  11.222 1.888  10.756 2.059  10.297 2.190  9.848 2.282  9.412 2.338  8.991 2.357  8.587 2.342  8.203 2.295  7.841 2.217  7.504 2.111  7.192 1.978  6.909 1.821  6.655 1.642  6.431 1.444  6.240 1.229  6.083 1.000  5.959 0.760  5.870 0.511  5.817 0.257  5.799 0.000  5.817 -0.257  5.870 -0.511  5.959 -0.760  6.083 -1.000  6.240 -1.229  6.431 -1.444  6.655 -1.642  6.909 -1.821  7.192 -1.978  7.504 -2.111  7.841 -2.217  8.203 -2.295  8.587 -2.342  8.991 -2.357  9.412 -2.338  9.848 -2.282  10.297 -2.190  10.756 -2.059  11.222 -1.888  11.691 -1.677  12.163 -1.425  12.632 -1.132  13.097 -0.796  13.554 -0.419  14.000 0.000 /

\color{black}
\plot 14.000 0.000  14.432 0.460  14.848 0.962  15.244 1.503  15.617 2.084  15.964 2.703  16.282 3.358  16.570 4.049  16.823 4.774  17.040 5.529  17.218 6.315  17.354 7.127  17.447 7.965  17.494 8.825  17.494 9.705  17.444 10.601  17.344 11.511  17.192 12.432  16.987 13.361  16.727 14.295  16.412 15.229  16.042 16.161  15.615 17.088  15.133 18.005  14.594 18.910  14.000 19.799  13.350 20.668  12.646 21.515  11.889 22.334  11.079 23.125  10.218 23.881  9.307 24.602  8.349 25.283  7.344 25.922  6.297 26.514  5.208 27.059  4.081 27.552  2.917 27.992  1.721 28.376  0.494 28.702  -0.760 28.967  -2.037 29.171  -3.334 29.310  -4.648 29.384  -5.975 29.392  -7.311 29.333  -8.653 29.204  -9.997 29.007  -11.339 28.741  -12.674 28.405  -14.000 28.000  -15.312 27.526  -16.606 26.983  -17.879 26.372  -19.126 25.695  -20.344 24.952  -21.529 24.146  -22.678 23.276  -23.786 22.347  -24.852 21.359  -25.870 20.315  -26.839 19.217  -27.754 18.069  -28.614 16.872  -29.416 15.631  -30.156 14.347  -30.833 13.025  -31.445 11.668  -31.990 10.280  -32.465 8.864  -32.870 7.423  -33.203 5.963  -33.463 4.487  -33.650 2.998  -33.762 1.501  -33.799 -0.000  -33.762 -1.501  -33.650 -2.998  -33.463 -4.487  -33.203 -5.963  -32.870 -7.423  -32.465 -8.864  -31.990 -10.280  -31.445 -11.668  -30.833 -13.025  -30.156 -14.347  -29.416 -15.631  -28.614 -16.872  -27.754 -18.069  -26.839 -19.217  -25.870 -20.315  -24.852 -21.359  -23.786 -22.347  -22.678 -23.276  -21.529 -24.146  -20.344 -24.952  -19.126 -25.695  -17.879 -26.372  -16.606 -26.983  -15.312 -27.526  -14.000 -28.000  -12.674 -28.405  -11.339 -28.741  -9.997 -29.007  -8.653 -29.204  -7.311 -29.333  -5.975 -29.392  -4.648 -29.384  -3.334 -29.310  -2.037 -29.171  -0.760 -28.967  0.494 -28.702  1.721 -28.376  2.917 -27.992  4.081 -27.552  5.208 -27.059  6.297 -26.514  7.344 -25.922  8.349 -25.283  9.307 -24.602  10.218 -23.881  11.079 -23.125  11.889 -22.334  12.646 -21.515  13.350 -20.668  14.000 -19.799  14.594 -18.910  15.133 -18.005  15.615 -17.088  16.042 -16.161  16.412 -15.229  16.727 -14.295  16.987 -13.361  17.192 -12.432  17.344 -11.511  17.444 -10.601  17.494 -9.705  17.494 -8.825  17.447 -7.965  17.354 -7.127  17.218 -6.315  17.040 -5.529  16.823 -4.774  16.570 -4.049  16.282 -3.358  15.964 -2.703  15.617 -2.084  15.244 -1.503  14.848 -0.962  14.432 -0.460  14.000 0.000   /

\color{Cyan} 
\multiput {$\bullet$} at -18.197 0.000  -14.579 -10.284  -14.579 10.284  -5.471 -15.825  -5.471 15.825  4.648 -13.999  4.648 13.999   /
\multiput {$\circ$} at  1.973 15.095  -7.732 16.937  -17.148 10.992  5.207 0.000  -17.148 -10.992  -7.732 -16.937  1.973 -15.095    /

\color{Orange}  
\multiput {$\bullet$} at -23.354 0.000  -22.059 -7.240  -22.059 7.240  -18.355 -13.532  -18.355 13.532  -12.758 -18.057  -12.758 18.057  -6.053 -20.232  -6.053 20.232  0.799 -19.788  0.799 19.788  6.772 -16.789  6.772 16.789  10.831 -11.557  10.831 11.557  /
\multiput {$\circ$} at  7.625 -3.342  6.815 16.126  0.539 18.933  -6.581 19.657  -13.273 17.702  -18.828 13.319  -22.502 7.138  5.429 0.000  -22.502 -7.138  -18.828 -13.319  -13.273 -17.702  -6.581 -19.657  0.539 -18.933  6.815 -16.126  7.625 3.342   /

\color{Red} 
\multiput {$\bullet$} at -27.272 0.000  -26.871 -4.398  -26.871 4.398  -25.683 -8.642  -25.683 8.642  -23.752 -12.585  -23.752 12.585  -21.150 -16.089  -21.150 16.089  -17.974 -19.033  -17.974 19.033  -14.343 -21.317  -14.343 21.317  -10.394 -22.863  -10.394 22.863  -6.277 -23.621  -6.277 23.621  -2.151 -23.569  -2.151 23.569  1.821 -22.712  1.821 22.712  5.477 -21.086  5.477 21.086  8.657 -18.750  8.657 18.750  11.203 -15.783  11.203 15.783  12.942 -12.266  12.942 12.266  13.631 -8.227  13.631 8.227  /
\multiput {$\circ$} at  9.215 -3.422  13.015 12.021  11.349 15.506  8.885 18.443  5.798 20.757  2.246 22.373  -1.613 23.235  -5.620 23.310  -9.613 22.591  -13.435 21.100  -16.938 18.886  -19.984 16.021  -22.451 12.599  -24.235 8.726  -25.263 4.501  -9.849 0.000  -25.263 -4.501  -24.235 -8.726  -22.451 -12.599  -19.984 -16.021  -16.938 -18.886  -13.435 -21.100  -9.613 -22.591  -5.620 -23.310  -1.613 -23.235  2.246 -22.373  5.798 -20.757  8.885 -18.443  11.349 -15.506  13.015 -12.021  9.215 3.422   /

\color{Blue} 
\multiput {$\bullet$} at -29.920 0.000  -29.807 -2.458  -29.807 2.458  -29.467 -4.893  -29.467 4.893  -28.905 -7.285  -28.905 7.285  -28.125 -9.611  -28.125 9.611  -27.136 -11.850  -27.136 11.850  -25.946 -13.982  -25.946 13.982  -24.567 -15.989  -24.567 15.989  -23.013 -17.852  -23.013 17.852  -21.298 -19.555  -21.298 19.555  -19.439 -21.082  -19.439 21.082  -17.453 -22.420  -17.453 22.420  -15.361 -23.557  -15.361 23.557  -13.183 -24.485  -13.183 24.485  -10.940 -25.194  -10.940 25.194  -8.653 -25.678  -8.653 25.678  -6.347 -25.935  -6.347 25.935  -4.043 -25.963  -4.043 25.963  -1.765 -25.762  -1.765 25.762  0.463 -25.334  0.463 25.334  2.619 -24.684  2.619 24.684  4.678 -23.820  4.678 23.820  6.619 -22.748  6.619 22.748  8.418 -21.480  8.418 21.480  10.054 -20.028  10.054 20.028  11.504 -18.403  11.504 18.403  12.744 -16.620  12.744 16.620  13.751 -14.691  13.751 14.691  14.495 -12.629  14.495 12.629  14.603 -5.589  14.603 5.589  14.939 -10.440  14.939 10.440  15.021 -8.115  15.021 8.115  /
\multiput {$\circ$} at   10.573 -3.075  9.645 -3.016  14.962 10.322  14.541 12.498  13.822 14.546  12.843 16.462  11.633 18.235  10.217 19.853  8.618 21.302  6.858 22.571  4.958 23.648  2.940 24.524  0.828 25.190  -1.357 25.641  -3.592 25.871  -5.854 25.878  -8.120 25.662  -10.369 25.224  -12.578 24.570  -14.726 23.703  -16.793 22.633  -18.759 21.368  -20.607 19.919  -22.318 18.301  -23.878 16.528  -25.273 14.614  -26.490 12.578  -27.521 10.436  -28.355 8.208  -28.986 5.912  -29.410 3.566  5.676 0.000  -29.410 -3.566  -28.986 -5.912  -28.355 -8.208  -27.521 -10.436  -26.490 -12.578  -25.273 -14.614  -23.878 -16.528  -22.318 -18.301  -20.607 -19.919  -18.759 -21.368  -16.793 -22.633  -14.726 -23.703  -12.578 -24.570  -10.369 -25.224  -8.120 -25.662  -5.854 -25.878  -3.592 -25.871  -1.357 -25.641  0.828 -25.190  2.940 -24.524  4.958 -23.648  6.858 -22.571  8.618 -21.302  10.217 -19.853  11.633 -18.235  12.843 -16.462  13.822 -14.546  14.541 -12.498  14.962 -10.322  9.645 3.016  10.573 3.075  /

\color{black}
\normalcolor

\endpicture
\end{center}
\caption{ Approximate ($\circ$) and exact zeros ($\bullet$) for 
$n\in \{8, 16, 32, 64\}$.  The approximate solutions are obtained by 
solving equation \Ref{eqn24} with $\beta=4$.}
\label{figure3}
\end{figure}

\begin{figure}[t]
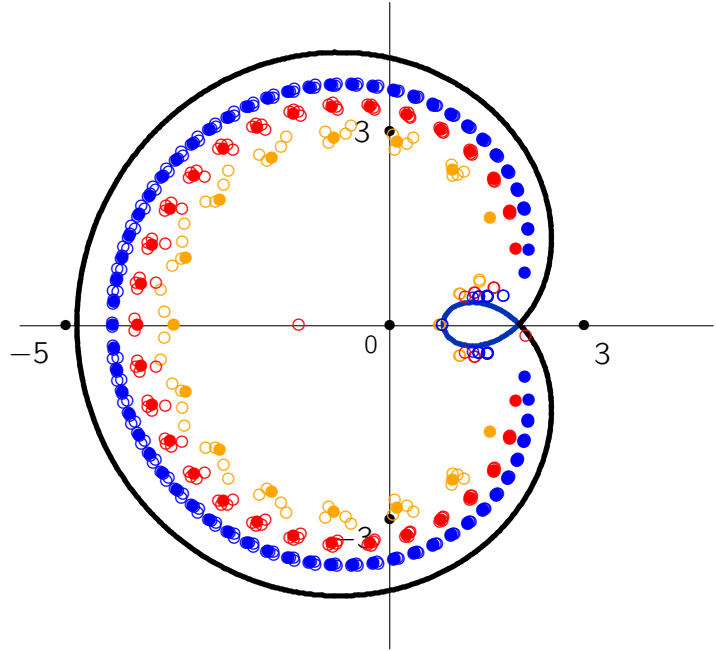

\begin{center}
\beginpicture
\color{black}
\setcoordinatesystem units <3.5pt,3.5pt>
\axes{-40}{-35}{0}{0}{35}{35}{-70}
\put {\footnotesize$0$} at -2 -2
\put {\footnotesize$\bullet$} at -35 0
\put {$-5$} at -39 -3
\put {\footnotesize$\bullet$} at 21 0
\put {$3$} at 23 -3
\put {\footnotesize$\bullet$} at 0 21
\put {$3$} at -3 21
\put {\footnotesize$\bullet$} at 0 -21
\put {$-3$} at -4 -23

\color{RoyalBlue}
\setplotsymbol ({\Large$\cdot$})
\plot 14.000 0.000  13.554 0.419  13.097 0.796  12.632 1.132  12.163 1.425  11.691 1.677  11.222 1.888  10.756 2.059  10.297 2.190  9.848 2.282  9.412 2.338  8.991 2.357  8.587 2.342  8.203 2.295  7.841 2.217  7.504 2.111  7.192 1.978  6.909 1.821  6.655 1.642  6.431 1.444  6.240 1.229  6.083 1.000  5.959 0.760  5.870 0.511  5.817 0.257  5.799 0.000  5.817 -0.257  5.870 -0.511  5.959 -0.760  6.083 -1.000  6.240 -1.229  6.431 -1.444  6.655 -1.642  6.909 -1.821  7.192 -1.978  7.504 -2.111  7.841 -2.217  8.203 -2.295  8.587 -2.342  8.991 -2.357  9.412 -2.338  9.848 -2.282  10.297 -2.190  10.756 -2.059  11.222 -1.888  11.691 -1.677  12.163 -1.425  12.632 -1.132  13.097 -0.796  13.554 -0.419  14.000 0.000 /

\color{black}
\plot 14.000 0.000  14.432 0.460  14.848 0.962  15.244 1.503  15.617 2.084  15.964 2.703  16.282 3.358  16.570 4.049  16.823 4.774  17.040 5.529  17.218 6.315  17.354 7.127  17.447 7.965  17.494 8.825  17.494 9.705  17.444 10.601  17.344 11.511  17.192 12.432  16.987 13.361  16.727 14.295  16.412 15.229  16.042 16.161  15.615 17.088  15.133 18.005  14.594 18.910  14.000 19.799  13.350 20.668  12.646 21.515  11.889 22.334  11.079 23.125  10.218 23.881  9.307 24.602  8.349 25.283  7.344 25.922  6.297 26.514  5.208 27.059  4.081 27.552  2.917 27.992  1.721 28.376  0.494 28.702  -0.760 28.967  -2.037 29.171  -3.334 29.310  -4.648 29.384  -5.975 29.392  -7.311 29.333  -8.653 29.204  -9.997 29.007  -11.339 28.741  -12.674 28.405  -14.000 28.000  -15.312 27.526  -16.606 26.983  -17.879 26.372  -19.126 25.695  -20.344 24.952  -21.529 24.146  -22.678 23.276  -23.786 22.347  -24.852 21.359  -25.870 20.315  -26.839 19.217  -27.754 18.069  -28.614 16.872  -29.416 15.631  -30.156 14.347  -30.833 13.025  -31.445 11.668  -31.990 10.280  -32.465 8.864  -32.870 7.423  -33.203 5.963  -33.463 4.487  -33.650 2.998  -33.762 1.501  -33.799 -0.000  -33.762 -1.501  -33.650 -2.998  -33.463 -4.487  -33.203 -5.963  -32.870 -7.423  -32.465 -8.864  -31.990 -10.280  -31.445 -11.668  -30.833 -13.025  -30.156 -14.347  -29.416 -15.631  -28.614 -16.872  -27.754 -18.069  -26.839 -19.217  -25.870 -20.315  -24.852 -21.359  -23.786 -22.347  -22.678 -23.276  -21.529 -24.146  -20.344 -24.952  -19.126 -25.695  -17.879 -26.372  -16.606 -26.983  -15.312 -27.526  -14.000 -28.000  -12.674 -28.405  -11.339 -28.741  -9.997 -29.007  -8.653 -29.204  -7.311 -29.333  -5.975 -29.392  -4.648 -29.384  -3.334 -29.310  -2.037 -29.171  -0.760 -28.967  0.494 -28.702  1.721 -28.376  2.917 -27.992  4.081 -27.552  5.208 -27.059  6.297 -26.514  7.344 -25.922  8.349 -25.283  9.307 -24.602  10.218 -23.881  11.079 -23.125  11.889 -22.334  12.646 -21.515  13.350 -20.668  14.000 -19.799  14.594 -18.910  15.133 -18.005  15.615 -17.088  16.042 -16.161  16.412 -15.229  16.727 -14.295  16.987 -13.361  17.192 -12.432  17.344 -11.511  17.444 -10.601  17.494 -9.705  17.494 -8.825  17.447 -7.965  17.354 -7.127  17.218 -6.315  17.040 -5.529  16.823 -4.774  16.570 -4.049  16.282 -3.358  15.964 -2.703  15.617 -2.084  15.244 -1.503  14.848 -0.962  14.432 -0.460  14.000 0.000   /

\color{Orange} 
\multiput {$\bullet$} at -23.354 0.000  -22.059 -7.240  -22.059 7.240  -18.355 -13.532  -18.355 13.532  -12.758 -18.057  -12.758 18.057  -6.053 -20.232  -6.053 20.232  0.799 -19.788  0.799 19.788  6.772 -16.789  6.772 16.789  10.831 -11.557  10.831 11.557   /
\multiput {$\circ$} at  7.625 -3.342  6.815 16.126  0.539 18.933  -6.581 19.657  -13.273 17.702  -18.828 13.319  -22.502 7.138  5.429 0.000  -22.502 -7.138  -18.828 -13.319  -13.273 -17.702  -6.581 -19.657  0.539 -18.933  6.815 -16.126  7.625 3.342    /
\multiput {$\circ$} at  8.042 16.606  2.541 20.240  -4.170 21.565  -11.153 20.341  -17.448 16.647  -22.193 10.898  -24.739 3.791  5.370 0.000  -24.739 -3.791  -22.193 -10.898  -17.448 -16.647  -11.153 -20.341  -4.170 -21.565  2.541 -20.240  8.042 -16.606   /
\multiput {$\circ$} at  7.338 16.242  1.882 19.583  -4.766 20.716  -11.734 19.253  -17.948 15.182  -22.404 8.961  -24.365 1.424  -23.490 -6.368  -19.884 -13.300  -14.082 -18.381  -6.968 -20.895  0.351 -20.512  6.690 -17.313  7.545 3.314  9.744 4.627  /
\multiput {$\circ$} at  7.545 -3.314  6.690 17.313  0.351 20.512  -6.968 20.895  -14.082 18.381  -19.884 13.300  -23.490 6.368  -24.365 -1.424  -22.404 -8.961  -17.948 -15.182  -11.734 -19.253  -4.766 -20.716  1.882 -19.583  7.338 -16.242  9.743 4.834   /

\color{Red} 
\multiput {$\bullet$} at -27.272 0.000  -26.871 -4.398  -26.871 4.398  -25.683 -8.642  -25.683 8.642  -23.752 -12.585  -23.752 12.585  -21.150 -16.089  -21.150 16.089  -17.974 -19.033  -17.974 19.033  -14.343 -21.317  -14.343 21.317  -10.394 -22.863  -10.394 22.863  -6.277 -23.621  -6.277 23.621  -2.151 -23.569  -2.151 23.569  1.821 -22.712  1.821 22.712  5.477 -21.086  5.477 21.086  8.657 -18.750  8.657 18.750  11.203 -15.783  11.203 15.783  12.942 -12.266  12.942 12.266  13.631 -8.227  13.631 8.227  /
\multiput {$\circ$} at  9.215 -3.422  13.015 12.021  11.349 15.506  8.885 18.443  5.798 20.757  2.246 22.373  -1.613 23.235  -5.620 23.310  -9.613 22.591  -13.435 21.100  -16.938 18.886  -19.984 16.021  -22.451 12.599  -24.235 8.726  -25.263 4.501  -9.849 0.000  -25.263 -4.501  -24.235 -8.726  -22.451 -12.599  -19.984 -16.021  -16.938 -18.886  -13.435 -21.100  -9.613 -22.591  -5.620 -23.310  -1.613 -23.235  2.246 -22.373  5.798 -20.757  8.885 -18.443  11.349 -15.506  13.015 -12.021  9.215 3.422   /
\multiput {$\circ$} at  13.037 12.177  8.148 -3.076  11.358 15.747  8.855 18.781  5.694 21.197  2.029 22.911  -1.981 23.853  -6.174 23.979  -10.380 23.273  -14.434 21.752  -18.176 19.461  -21.460 16.479  -24.157 12.907  -26.162 8.871  -27.397 4.517  -27.397 -4.517  -26.162 -8.871  -24.157 -12.907  -21.460 -16.479  -18.176 -19.461  -14.434 -21.752  -10.380 -23.273  -6.174 -23.979  -1.981 -23.853  2.029 -22.911  5.694 -21.197  8.855 -18.781  11.358 -15.747  8.148 3.076  13.037 -12.177  11.288 3.967   /
\multiput {$\circ$} at  12.930 12.102  11.238 15.590  8.765 18.527  5.685 20.856  2.141 22.518  -1.727 23.458  -5.781 23.635  -9.872 23.027  -13.847 21.636  -17.553 19.493  -20.844 16.660  -23.587 13.229  -25.673 9.315  -27.015 5.055  -27.556 0.603  -27.272 -3.884  -26.169 -8.240  -24.289 -12.309  -21.701 -15.943  -18.506 -19.010  -14.825 -21.401  -10.803 -23.030  -6.598 -23.844  -2.378 -23.817  1.685 -22.958  5.420 -21.302  8.659 -18.916  11.239 -15.882  12.980 -12.286  9.163 3.421  11.261 3.957  /
\multiput {$\circ$} at  9.163 -3.421  12.980 12.286  11.239 15.882  8.659 18.916  5.420 21.302  1.685 22.958  -2.378 23.817  -6.598 23.844  -10.803 23.030  -14.825 21.401  -18.506 19.010  -21.701 15.943  -24.289 12.309  -26.169 8.240  -27.272 3.884  -27.556 -0.603  -27.015 -5.055  -25.673 -9.315  -23.587 -13.229  -20.844 -16.660  -17.553 -19.493  -13.847 -21.636  -9.872 -23.027  -5.781 -23.635  -1.727 -23.458  2.141 -22.518  5.685 -20.856  8.765 -18.527  11.238 -15.590  12.930 -12.102  14.716 -1.207  /

\color{Blue} 
\multiput {$\bullet$} at -29.920 0.000  -29.807 -2.458  -29.807 2.458  -29.467 -4.893  -29.467 4.893  -28.905 -7.285  -28.905 7.285  -28.125 -9.611  -28.125 9.611  -27.136 -11.850  -27.136 11.850  -25.946 -13.982  -25.946 13.982  -24.567 -15.989  -24.567 15.989  -23.013 -17.852  -23.013 17.852  -21.298 -19.555  -21.298 19.555  -19.439 -21.082  -19.439 21.082  -17.453 -22.420  -17.453 22.420  -15.361 -23.557  -15.361 23.557  -13.183 -24.485  -13.183 24.485  -10.940 -25.194  -10.940 25.194  -8.653 -25.678  -8.653 25.678  -6.347 -25.935  -6.347 25.935  -4.043 -25.963  -4.043 25.963  -1.765 -25.762  -1.765 25.762  0.463 -25.334  0.463 25.334  2.619 -24.684  2.619 24.684  4.678 -23.820  4.678 23.820  6.619 -22.748  6.619 22.748  8.418 -21.480  8.418 21.480  10.054 -20.028  10.054 20.028  11.504 -18.403  11.504 18.403  12.744 -16.620  12.744 16.620  13.751 -14.691  13.751 14.691  14.495 -12.629  14.495 12.629  14.603 -5.589  14.603 5.589  14.939 -10.440  14.939 10.440  15.021 -8.115  15.021 8.115  /
\multiput {$\circ$} at  10.573 -3.075  9.645 -3.016  14.962 10.322  14.541 12.498  13.822 14.546  12.843 16.462  11.633 18.235  10.217 19.853  8.618 21.302  6.858 22.571  4.958 23.648  2.940 24.524  0.828 25.190  -1.357 25.641  -3.592 25.871  -5.854 25.878  -8.120 25.662  -10.369 25.224  -12.578 24.570  -14.726 23.703  -16.793 22.633  -18.759 21.368  -20.607 19.919  -22.318 18.301  -23.878 16.528  -25.273 14.614  -26.490 12.578  -27.521 10.436  -28.355 8.208  -28.986 5.912  -29.410 3.566  5.676 0.000  -29.410 -3.566  -28.986 -5.912  -28.355 -8.208  -27.521 -10.436  -26.490 -12.578  -25.273 -14.614  -23.878 -16.528  -22.318 -18.301  -20.607 -19.919  -18.759 -21.368  -16.793 -22.633  -14.726 -23.703  -12.578 -24.570  -10.369 -25.224  -8.120 -25.662  -5.854 -25.878  -3.592 -25.871  -1.357 -25.641  0.828 -25.190  2.940 -24.524  4.958 -23.648  6.858 -22.571  8.618 -21.302  10.217 -19.853  11.633 -18.235  12.843 -16.462  13.822 -14.546  14.541 -12.498  14.962 -10.322  9.645 3.016  10.573 3.075    /
\multiput {$\circ$} at  10.579 -3.074  9.651 -3.017  9.004 -2.925  14.584 12.472  13.880 14.522  12.919 16.443  11.726 18.225  10.327 19.857  8.743 21.325  6.996 22.616  5.105 23.719  3.093 24.624  0.981 25.321  -1.208 25.805  -3.453 26.068  -5.732 26.109  -8.020 25.925  -10.298 25.517  -12.541 24.889  -14.728 24.044  -16.839 22.990  -18.854 21.735  -20.752 20.289  -22.515 18.666  -24.128 16.878  -25.574 14.941  -26.841 12.873  -27.915 10.692  -28.787 8.415  -29.448 6.065  -29.892 3.660  5.673 0.000  -29.892 -3.660  -29.448 -6.065  -28.787 -8.415  -27.915 -10.692  -26.841 -12.873  -25.574 -14.941  -24.128 -16.878  -22.515 -18.666  -20.752 -20.289  -18.854 -21.735  -16.839 -22.990  -14.728 -24.044  -12.541 -24.889  -10.298 -25.517  -8.020 -25.925  -5.732 -26.109  -3.453 -26.068  -1.208 -25.805  0.981 -25.321  3.093 -24.624  5.105 -23.719  6.996 -22.616  8.743 -21.325  10.327 -19.857  11.726 -18.225  12.919 -16.443  13.880 -14.522  14.584 -12.472  9.004 2.925  9.651 3.017  10.579 3.074    /
\multiput {$\circ$} at  10.550 -3.089  9.617 -3.025  14.926 10.397  14.486 12.586  13.745 14.644  12.743 16.567  11.510 18.343  10.070 19.962  8.446 21.409  6.660 22.673  4.736 23.743  2.693 24.609  0.556 25.264  -3.916 25.915  -6.205 25.903  -8.499 25.666  -10.777 25.202  -13.015 24.517  -15.192 23.615  -17.287 22.502  -19.279 21.189  -21.150 19.685  -22.880 18.004  -24.452 16.161  -25.852 14.170  -27.065 12.050  -28.080 9.819  -28.885 7.499  -29.473 5.108  -29.839 2.670  -29.977 0.205  -29.887 -2.262  -29.569 -4.711  -29.025 -7.118  -28.262 -9.462  -27.285 -11.721  -26.105 -13.875  -24.733 -15.903  -23.182 -17.788  -21.467 -19.512  -19.605 -21.060  -17.614 -22.418  -15.514 -23.574  -13.324 -24.517  -11.068 -25.240  -8.768 -25.736  -6.446 -26.002  -4.127 -26.036  -1.833 -25.838  0.410 -25.411  2.580 -24.760  4.654 -23.891  6.607 -22.814  8.416 -21.538  10.061 -20.076  11.517 -18.441  12.761 -16.647  13.769 -14.706  14.512 -12.631  14.950 -10.428  9.618 3.019  10.549 3.082  12.282 3.122  /
\multiput {$\circ$} at  10.549 -3.082  9.618 -3.019  14.950 10.428  14.512 12.631  13.769 14.706  12.761 16.647  11.517 18.441  10.061 20.076  8.416 21.538  6.607 22.814  4.654 23.891  2.580 24.760  0.410 25.411  -1.833 25.838  -4.127 26.036  -6.446 26.002  -8.768 25.736  -11.068 25.240  -13.324 24.517  -15.514 23.574  -17.614 22.418  -19.605 21.060  -21.467 19.512  -23.182 17.788  -24.733 15.903  -26.105 13.875  -27.285 11.721  -28.262 9.462  -29.025 7.118  -29.569 4.711  -29.887 2.262  -29.977 -0.205  -29.839 -2.670  -29.473 -5.108  -28.885 -7.499  -28.080 -9.819  -27.065 -12.050  -25.852 -14.170  -24.452 -16.161  -22.880 -18.004  -21.150 -19.685  -19.279 -21.189  -17.287 -22.502  -15.192 -23.615  -13.015 -24.517  -10.777 -25.202  -8.499 -25.666  -6.205 -25.903  -3.916 -25.915  0.556 -25.264  2.693 -24.609  4.736 -23.743  6.660 -22.673  8.446 -21.409  10.070 -19.962  11.510 -18.343  12.743 -16.567  13.745 -14.644  14.486 -12.586  14.926 -10.397  9.617 3.025  10.550 3.089  12.287 3.129   /

\color{black}
\normalcolor

\endpicture
\end{center}
\caption{ Approximate ($\circ$) and exact zeros ($\bullet$) for $n=16$, $n=32$ and $n=64$ and
for $\beta\in\{\pm 4, \pm 4\Imi\}$.  For $n=32$ and $n=64$ the zeros and estimated zeros cluster
together, but less so for $n=16$.  Notice that the numerical procedure finds, in some cases,
solutions on the other branch of the lima\c{c}on.}
\label{figure4}
\end{figure}

In Figure \ref{figure4} the effects of $\beta$ on the location of approximate zeros are
examined.  The exact zeros for $n=16$ and $n=32$ are plotted (denoted by bullets).
Approximations for $\beta\in\{\pm 4, \pm 4\Imi\}$ are shown as open circles.  For $n=32$
the estimates cluster close to the exact values.  This is less so for $n=16$, but even in this
case the effects of $\beta$ are damped down in equation \Ref{eqn24}.  Thus, we will 
restrict our attention to the case $\beta=0$, and work on finding approximate solutions to $D^0_{2n}(a) = 0$.

Put $\beta=0$ in equation \Ref{eqn24} to obtain
\begin{equation}
a^2 = 4(a-1) e^{(2k+1)\pi\Imi/n} \L \frac{1}{\sqrt{\pi n^3}} \R^{\! 1/n}\!
\L \frac{a(a-1)}{(a-2)^3} \R^{\! 1/n}.
\label{eqn25}  
\end{equation}
Define the functions
\begin{equation}
\sigma_{k,n} = e^{(2k+1)\pi\Imi/n},\q
h_n =  \L \frac{1}{\sqrt{\pi n^3}} \R^{\! 1/n},\q
H_n(a) = \L \frac{a(a-1)}{(a-2)^3} \R^{\! 1/n} .
\end{equation}
Use these definitions to write equation \Ref{eqn25} in the following form:
\begin{equation}
a^2 = 4(a-1) \sigma_{k,n} \, h_n \, H_n(a) .
\label{eqn27}  
\end{equation}
If it was the case that $H_n(a)$ did not depend on $a$, then this equation can
be solved for $a$.  Thus, our strategy is to find an approximation for $a$
so that $H_n(a)$ can be approximated as a function of $(k,n)$, and then to solve
the resulting equation for $a$.  Since $H_n(a) \to 1$ as $n\to\infty$, replace it by $1$
in equation \Ref{eqn27} to find the reduced equation
\begin{equation}
a^2 = 4(a-1) \sigma_{k,n} \, h_n
\label{eqn28}  
\end{equation}
for approximate zeros.  The solutions of this equation are
\begin{equation}
a'_\pm = 2\sigma_{k,n} h_n \pm 2 \sqrt{\sigma_{k,n}h_n}\; \sqrt{\sigma_{k,n} h_n -  1 } .
\end{equation}
Substituting this in $H_n(a)$ in equation \Ref{eqn27} gives the following asymptotic
expression for $H_n(a_0)$:
\begin{equation}
\hspace{-2.5cm}
H_n(a'_\pm) \sim 1 + \sfrac{1}{n} \log \L 
\frac{(\sigma_{k,n}\pm \sqrt{\sigma_{k,n}}\; \sqrt{\sigma_{k,n}-1}\, )
 (2\sigma_{k,n}\pm2 \sqrt{\sigma_{k,n}}\; \sqrt{\sigma_{k,n}-1}\,-1 )
}{4\L \sigma_{k,n}-1\pm \sqrt{\sigma_{k,n}}\; \sqrt{(\sigma_{k,n}-1) } \R^3\,} \R\!\! .
\label{eqn30}  
\end{equation}
Proceed by substituting $H_n(a)$ in equation \Ref{eqn27} with this asymptotic
expression, and denote the solutions by $a''$.  There are two choices of the
sign in equation \Ref{eqn30}, and in total four solutions are found. However, 
these are in identical pairs so that only two distinct solutions are obtained.  
Asymptotic expansions of these two solutions are
\begin{eqnarray}
\hspace{-2.5cm}
a''_\pm 
& \hspace{-2.0cm}
\simeq 2 \sigma_{k,n} \pm 2  \sqrt{\sigma_{k,n}}\; \sqrt{\sigma_{k,n}-1}\, \nonumber \\
& \hspace{-1.5cm}
\mp \Sfrac{1}{n} \L \Sfrac{3}{2} \log n + 2\log 2 + \Sfrac{1}{2} \log \pi
        \R \frac{\sqrt{\sigma_{k,n}}\, (2 \sigma_{k,n} 
              \pm 2  \sqrt{\sigma_{k,n}}\; \sqrt{\sigma_{k,n}-1}\, -1 )}{\sqrt{\sigma_{k,n}-1}}
                   \nonumber \\
& \hspace{-1.0cm}
\pm \Sfrac{1}{n} \frac{\sqrt{\sigma_{k.n}}\, (2 \sigma_{k,n} 
              \pm 2  \sqrt{\sigma_{k,n}}\; \sqrt{\sigma_{k,n}-1}\, -1 )}{\sqrt{\sigma_{k,n}-1}}
                 \times \nonumber \\
& \hspace{-0.5cm} \log \L
\frac{(\sigma_{k,n} +  \sqrt{\sigma_{k,n}}\; \sqrt{\sigma_{k,n}-1}\,)(2 \sigma_{k,n} 
              + 2 \sqrt{\sigma_{k,n}}\; \sqrt{\sigma_{k,n}-1}\, -1 )}
{\L \sigma_{k,n} +  \sqrt{\sigma_{k,n}}\; \sqrt{\sigma_{k,n}-1}\,-1 \R^3}\R  .
\label{eqn31}  
\end{eqnarray}

\begin{figure}[t]
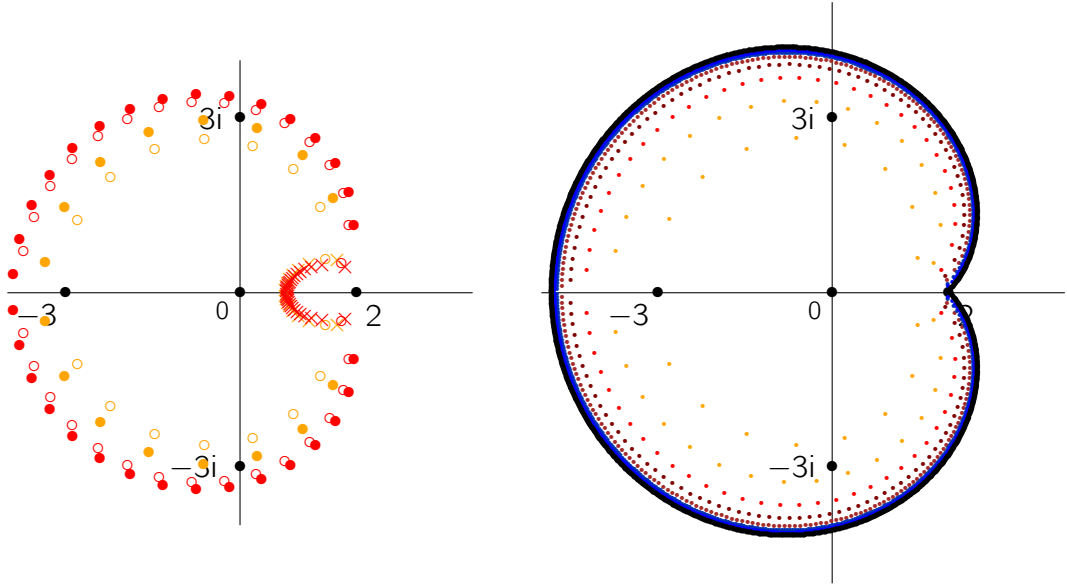

\begin{center}
\beginpicture
\color{black}
\setcoordinatesystem units <2.2pt,2.2pt>
\setplotarea x from -100 to 100, y from -50 to 50
\axes{-90}{-40}{-50}{0}{-10}{40}{-110}
\put {\footnotesize$0$} at -53 -3
\put {\footnotesize$\bullet$} at -80 0
\put {$-3$} at -85 -4
\put {\footnotesize$\bullet$} at -30 0 
\put {$2$} at -27 -4
\put {\footnotesize$\bullet$} at -50 30
\put {$3\Imi$} at -55 30
\put {\footnotesize$\bullet$} at -50 -30
\put {$-3\Imi$} at -58 -30

\color{Orange} 
\multiput {\footnotesize$\circ$} at  -35.290 5.710  -36.181 14.559  -40.814 21.047  -47.827 25.135  -56.143 26.358  -64.654 24.504  -72.260 19.699  -77.965 12.421  -77.965 -12.421  -72.260 -19.699  -64.654 -24.504  -56.143 -26.358  -47.827 -25.135  -40.814 -21.047  -36.181 -14.559  -35.290 -5.710  /
\multiput {\footnotesize$\times$} at  -33.312 -5.596  -38.104 -4.904  -39.930 -3.935  -40.908 -3.096  -41.498 -2.364  -41.869 -1.708  -42.100 -1.101  -42.227 -0.522  -42.227 0.522  -42.100 1.101  -41.869 1.708  -41.498 2.364  -40.908 3.096  -39.930 3.935  -38.104 4.904  -33.312 5.596  /
\multiput {\footnotesize$\bullet$} at  -83.506 -5.048  -83.506 5.048  -80.204 -14.554  -80.204 14.554  -74.009 -22.362  -74.009 22.362  -65.695 -27.566  -65.695 27.566  -56.308 -29.574  -56.308 29.574  -47.059 -28.173  -47.059 28.173  -39.205 -23.532  -39.205 23.532  -34.015 -16.087  -34.015 16.087  /
\color{Red} 
\multiput {\footnotesize$\circ$} at -32.554 4.957  -31.312 11.522  -32.245 16.960  -34.577 21.746  -38.002 25.812  -42.297 29.040  -47.251 31.320  -52.652 32.569  -58.284 32.737  -63.930 31.809  -69.376 29.809  -74.420 26.799  -78.875 22.875  -82.578 18.166  -85.392 12.830  -87.212 7.045  -87.212 -7.045  -85.392 -12.830  -82.578 -18.166  -78.875 -22.875  -74.420 -26.799  -69.376 -29.809  -63.930 -31.809  -58.284 -32.737  -52.652 -32.569  -47.251 -31.320  -42.297 -29.040  -38.002 -25.812  -34.577 -21.746  -32.245 -16.960  -31.312 -11.522  -32.554 -4.957 /
\multiput {\footnotesize$\times$} at  -31.919 -4.448  -35.891 -4.747  -37.714 -4.416  -38.842 -4.019  -39.622 -3.624  -40.194 -3.247  -40.630 -2.889  -40.969 -2.550  -41.237 -2.227  -41.451 -1.919  -41.620 -1.623  -41.754 -1.336  -41.857 -1.058  -41.934 -0.785  -41.987 -0.517  -42.019 -0.251  -42.019 0.251  -41.987 0.517  -41.934 0.785  -41.857 1.058  -41.754 1.336  -41.620 1.623  -41.451 1.919  -41.237 2.227  -40.969 2.550  -40.630 2.889  -40.194 3.247  -39.622 3.624  -38.842 4.019  -37.714 4.416  -35.891 4.747  -31.919 4.448  /
\multiput {\footnotesize$\bullet$} at -89.028 -3.078  -89.028 3.078  -87.948 -9.133  -87.948 9.133  -85.825 -14.888  -85.825 14.888  -82.735 -20.154  -82.735 20.154  -78.785 -24.758  -78.785 24.758  -74.114 -28.552  -74.114 28.552  -68.888 -31.413  -68.888 31.413  -63.293 -33.250  -63.293 33.250  -57.530 -34.009  -57.530 34.009  -51.809 -33.671  -51.809 33.671  -46.345 -32.252  -46.345 32.252  -41.350 -29.805  -41.350 29.805  -37.033 -26.414  -37.033 26.414  -33.603 -22.184  -33.603 22.184  -31.289 -17.223  -31.289 17.223  -30.425 -11.559  -30.425 11.559  /

\color{black}
\normalcolor

\axes{0}{-50}{50}{0}{90}{50}{0}
\put {\footnotesize$0$} at 47 -3
\put {\footnotesize$\bullet$} at 20 0
\put {$-3$} at 15 -4
\put {\footnotesize$\bullet$} at 70 0
\put {$2$} at 73 -4
\put {\footnotesize$\bullet$} at 50 30
\put {$3\Imi$} at 45 30
\put {\footnotesize$\bullet$} at 50 -30
\put {$-3\Imi$} at 43 -30

\color{Orange} 
\multiput {\large$\cdot$} at 64.710 5.710  63.819 14.559  59.186 21.047  52.173 25.135  43.857 26.358  35.346 24.504  27.740 19.699  22.035 12.421  22.035 -12.421  27.740 -19.699  35.346 -24.504  43.857 -26.358  52.173 -25.135  59.186 -21.047  63.819 -14.559  64.710 -5.710   /
\color{Orange} 
\multiput {\large$\cdot$} at 67.446 4.957  68.688 11.522  67.755 16.960  65.423 21.746  61.998 25.812  57.703 29.040  52.749 31.320  47.348 32.569  41.716 32.737  36.070 31.809  30.624 29.809  25.580 26.799  21.125 22.875  17.422 18.166  14.608 12.830  12.788 7.045  12.788 -7.045  14.608 -12.830  17.422 -18.166  21.125 -22.875  25.580 -26.799  30.624 -29.809  36.070 -31.809  41.716 -32.737  47.348 -32.569  52.749 -31.320  57.703 -29.040  61.998 -25.812  65.423 -21.746  67.755 -16.960  68.688 -11.522  67.446 -4.957  /
\color{red} 
\multiput {\large$\cdot$} at 68.806 3.749  70.615 8.024  71.168 11.538  71.051 14.799  70.431 17.877  69.390 20.779  67.978 23.495  66.238 26.009  64.202 28.301  61.904 30.353  59.374 32.145  56.645 33.662  53.747 34.889  50.713 35.815  47.575 36.429  44.365 36.725  41.118 36.699  37.864 36.351  34.637 35.682  31.469 34.698  28.390 33.406  25.431 31.817  22.620 29.946  19.984 27.807  17.550 25.420  15.339 22.805  13.374 19.986  11.673 16.988  10.253 13.837  9.127 10.562  8.307 7.191  7.799 3.756  7.799 -3.756  8.307 -7.191  9.127 -10.562  10.253 -13.837  11.673 -16.988  13.374 -19.986  15.339 -22.805  17.550 -25.420  19.984 -27.807  22.620 -29.946  25.431 -31.817  28.390 -33.406  31.469 -34.698  34.637 -35.682  37.864 -36.351  41.118 -36.699  44.365 -36.725  47.575 -36.429  50.713 -35.815  53.747 -34.889  56.645 -33.662  59.374 -32.145  61.904 -30.353  64.202 -28.301  66.238 -26.009  67.978 -23.495  69.390 -20.779  70.431 -17.877  71.051 -14.799  71.168 -11.538  70.615 -8.024  68.806 -3.749  /
\color{Maroon} 
\multiput {\large$\cdot$}  at 69.450 2.709  71.109 5.419  71.961 7.551  72.452 9.511  72.704 11.375  72.769 13.174  72.677 14.919  72.447 16.615  72.091 18.266  71.619 19.869  71.040 21.426  70.359 22.933  69.583 24.388  68.717 25.790  67.764 27.135  66.731 28.421  65.621 29.645  64.439 30.805  63.189 31.897  61.875 32.920  60.501 33.871  59.073 34.747  57.593 35.548  56.067 36.270  54.498 36.912  52.892 37.473  51.252 37.951  49.583 38.344  47.889 38.652  46.175 38.874  44.446 39.009  42.705 39.057  40.958 39.017  39.208 38.890  37.461 38.675  35.720 38.374  33.990 37.985  32.276 37.511  30.581 36.953  28.911 36.310  27.268 35.586  25.657 34.781  24.083 33.896  22.548 32.935  21.057 31.899  19.614 30.791  18.222 29.612  16.884 28.366  15.604 27.055  14.385 25.683  13.229 24.252  12.140 22.766  11.121 21.227  10.173 19.640  9.299 18.009  8.502 16.336  7.783 14.625  7.143 12.882  6.585 11.108  6.111 9.310  5.720 7.489  5.414 5.652  5.193 3.802  5.059 1.943  5.059 -1.943  5.193 -3.802  5.414 -5.652  5.720 -7.489  6.111 -9.310  6.585 -11.108  7.143 -12.882  7.783 -14.625  8.502 -16.336  9.299 -18.009  10.173 -19.640  11.121 -21.227  12.140 -22.766  13.229 -24.252  14.385 -25.683  15.604 -27.055  16.884 -28.366  18.222 -29.612  19.614 -30.791  21.057 -31.899  22.548 -32.935  24.083 -33.896  25.657 -34.781  27.268 -35.586  28.911 -36.310  30.581 -36.953  32.276 -37.511  33.990 -37.985  35.720 -38.374  37.461 -38.675  39.208 -38.890  40.958 -39.017  42.705 -39.057  44.446 -39.009  46.175 -38.874  47.889 -38.652  49.583 -38.344  51.252 -37.951  52.892 -37.473  54.498 -36.912  56.067 -36.270  57.593 -35.548  59.073 -34.747  60.501 -33.871  61.875 -32.920  63.189 -31.897  64.439 -30.805  65.621 -29.645  66.731 -28.421  67.764 -27.135  68.717 -25.790  69.583 -24.388  70.359 -22.933  71.040 -21.426  71.619 -19.869  72.091 -18.266  72.447 -16.615  72.677 -14.919  72.769 -13.174  72.704 -11.375  72.452 -9.511  71.961 -7.551  71.109 -5.419  69.450 -2.709  /
\color{Brown} 
\multiput {\large$\cdot$} at 69.752 1.925  71.075 3.657  71.842 4.955  72.385 6.119  72.794 7.212  73.105 8.260  73.341 9.274  73.513 10.264  73.630 11.232  73.698 12.183  73.722 13.118  73.704 14.039  73.647 14.946  73.555 15.841  73.428 16.723  73.268 17.594  73.077 18.452  72.855 19.298  72.604 20.132  72.325 20.953  72.019 21.763  71.686 22.559  71.327 23.342  70.943 24.112  70.535 24.868  70.103 25.610  69.648 26.338  69.171 27.052  68.671 27.750  68.151 28.433  67.610 29.100  67.049 29.751  66.468 30.386  65.868 31.004  65.250 31.605  64.614 32.188  63.960 32.754  63.290 33.302  62.604 33.831  61.902 34.342  61.184 34.834  60.453 35.306  59.707 35.759  58.948 36.192  58.176 36.605  57.392 36.997  56.596 37.369  55.788 37.720  54.971 38.050  54.143 38.359  53.305 38.646  52.459 38.912  51.605 39.156  50.743 39.378  49.873 39.578  48.997 39.755  48.116 39.911  47.228 40.043  46.336 40.154  45.440 40.241  44.540 40.306  43.637 40.348  42.732 40.367  41.825 40.364  40.916 40.337  40.007 40.288  39.098 40.216  38.189 40.121  37.281 40.003  36.375 39.863  35.471 39.700  34.570 39.514  33.672 39.306  32.778 39.075  31.889 38.822  31.005 38.547  30.126 38.250  29.253 37.931  28.388 37.591  27.529 37.228  26.678 36.845  25.836 36.441  25.002 36.015  24.178 35.569  23.363 35.103  22.559 34.617  21.766 34.110  20.984 33.585  20.214 33.039  19.457 32.475  18.712 31.893  17.980 31.292  17.263 30.672  16.559 30.036  15.870 29.382  15.196 28.711  14.538 28.024  13.895 27.320  13.268 26.601  12.659 25.867  12.066 25.117  11.490 24.353  10.933 23.575  10.393 22.784  9.872 21.979  9.369 21.162  8.886 20.332  8.422 19.491  7.977 18.639  7.553 17.775  7.148 16.902  6.764 16.019  6.401 15.126  6.058 14.225  5.737 13.316  5.437 12.399  5.158 11.474  4.901 10.543  4.665 9.606  4.452 8.664  4.260 7.716  4.091 6.764  3.944 5.808  3.820 4.849  3.717 3.887  3.638 2.923  3.581 1.957  3.546 0.990  3.546 -0.990  3.581 -1.957  3.638 -2.923  3.717 -3.887  3.820 -4.849  3.944 -5.808  4.091 -6.764  4.260 -7.716  4.452 -8.664  4.665 -9.606  4.901 -10.543  5.158 -11.474  5.437 -12.399  5.737 -13.316  6.058 -14.225  6.401 -15.126  6.764 -16.019  7.148 -16.902  7.553 -17.775  7.977 -18.639  8.422 -19.491  8.886 -20.332  9.369 -21.162  9.872 -21.979  10.393 -22.784  10.933 -23.575  11.490 -24.353  12.066 -25.117  12.659 -25.867  13.268 -26.601  13.895 -27.320  14.538 -28.024  15.196 -28.711  15.870 -29.382  16.559 -30.036  17.263 -30.672  17.980 -31.292  18.712 -31.893  19.457 -32.475  20.214 -33.039  20.984 -33.585  21.766 -34.110  22.559 -34.617  23.363 -35.103  24.178 -35.569  25.002 -36.015  25.836 -36.441  26.678 -36.845  27.529 -37.228  28.388 -37.591  29.253 -37.931  30.126 -38.250  31.005 -38.547  31.889 -38.822  32.778 -39.075  33.672 -39.306  34.570 -39.514  35.471 -39.700  36.375 -39.863  37.281 -40.003  38.189 -40.121  39.098 -40.216  40.007 -40.288  40.916 -40.337  41.825 -40.364  42.732 -40.367  43.637 -40.348  44.540 -40.306  45.440 -40.241  46.336 -40.154  47.228 -40.043  48.116 -39.911  48.997 -39.755  49.873 -39.578  50.743 -39.378  51.605 -39.156  52.459 -38.912  53.305 -38.646  54.143 -38.359  54.971 -38.050  55.788 -37.720  56.596 -37.369  57.392 -36.997  58.176 -36.605  58.948 -36.192  59.707 -35.759  60.453 -35.306  61.184 -34.834  61.902 -34.342  62.604 -33.831  63.290 -33.302  63.960 -32.754  64.614 -32.188  65.250 -31.605  65.868 -31.004  66.468 -30.386  67.049 -29.751  67.610 -29.100  68.151 -28.433  68.671 -27.750  69.171 -27.052  69.648 -26.338  70.103 -25.610  70.535 -24.868  70.943 -24.112  71.327 -23.342  71.686 -22.559  72.019 -21.763  72.325 -20.953  72.604 -20.132  72.855 -19.298  73.077 -18.452  73.268 -17.594  73.428 -16.723  73.555 -15.841  73.647 -14.946  73.704 -14.039  73.722 -13.118  73.698 -12.183  73.630 -11.232  73.513 -10.264  73.341 -9.274  73.105 -8.260  72.794 -7.212  72.385 -6.119  71.842 -4.955  71.075 -3.657  69.752 -1.925  /
\color{RoyalBlue} 
\multiput {\large$\cdot$} at 69.892 1.360  70.888 2.486  71.489 3.295  71.940 4.002  72.305 4.655  72.612 5.273  72.874 5.866  73.101 6.439  73.300 6.997  73.473 7.543  73.625 8.078  73.757 8.605  73.871 9.124  73.969 9.636  74.052 10.142  74.121 10.642  74.176 11.138  74.219 11.629  74.250 12.115  74.270 12.597  74.279 13.075  74.277 13.550  74.265 14.021  74.243 14.488  74.212 14.952  74.171 15.413  74.122 15.871  74.063 16.325  73.996 16.776  73.921 17.224  73.837 17.669  73.746 18.111  73.646 18.550  73.539 18.985  73.424 19.418  73.301 19.847  73.172 20.273  73.035 20.696  72.891 21.116  72.740 21.533  72.582 21.946  72.417 22.356  72.245 22.763  72.067 23.166  71.883 23.567  71.692 23.963  71.495 24.357  71.291 24.746  71.081 25.133  70.866 25.515  70.644 25.894  70.416 26.270  70.183 26.641  69.944 27.009  69.699 27.374  69.448 27.734  69.192 28.090  68.931 28.443  68.664 28.792  68.392 29.136  68.114 29.477  67.832 29.813  67.544 30.146  67.251 30.474  66.954 30.798  66.651 31.117  66.344 31.432  66.032 31.743  65.715 32.050  65.394 32.352  65.068 32.650  64.738 32.943  64.404 33.231  64.065 33.515  63.722 33.794  63.374 34.069  63.023 34.339  62.668 34.604  62.308 34.864  61.945 35.120  61.578 35.370  61.207 35.616  60.833 35.856  60.455 36.092  60.073 36.322  59.688 36.548  59.300 36.768  58.908 36.984  58.513 37.194  58.115 37.399  57.714 37.598  57.310 37.793  56.903 37.982  56.493 38.166  56.080 38.344  55.664 38.517  55.246 38.685  54.825 38.847  54.402 39.004  53.976 39.155  53.548 39.301  53.118 39.442  52.685 39.576  52.251 39.705  51.814 39.829  51.375 39.947  50.934 40.059  50.492 40.166  50.047 40.267  49.601 40.362  49.154 40.452  48.705 40.536  48.254 40.614  47.802 40.687  47.349 40.753  46.894 40.814  46.438 40.869  45.981 40.919  45.523 40.962  45.065 41.000  44.605 41.032  44.144 41.058  43.683 41.078  43.221 41.092  42.759 41.101  42.296 41.104  41.833 41.101  41.369 41.092  40.905 41.077  40.441 41.056  39.977 41.030  39.512 40.997  39.048 40.959  38.584 40.915  38.120 40.865  37.656 40.809  37.193 40.748  36.730 40.680  36.268 40.607  35.806 40.528  35.344 40.444  34.884 40.353  34.424 40.257  33.965 40.155  33.507 40.047  33.050 39.933  32.594 39.814  32.139 39.689  31.685 39.558  31.233 39.422  30.781 39.280  30.332 39.132  29.884 38.979  29.437 38.820  28.992 38.656  28.549 38.486  28.107 38.311  27.668 38.130  27.230 37.943  26.794 37.751  26.361 37.554  25.929 37.351  25.500 37.143  25.073 36.930  24.648 36.711  24.225 36.487  23.805 36.258  23.388 36.023  22.973 35.784  22.561 35.539  22.151 35.289  21.745 35.034  21.341 34.774  20.940 34.509  20.542 34.239  20.147 33.964  19.755 33.684  19.366 33.399  18.980 33.110  18.598 32.815  18.219 32.516  17.844 32.213  17.471 31.904  17.103 31.591  16.738 31.274  16.376 30.952  16.019 30.625  15.665 30.294  15.314 29.959  14.968 29.619  14.625 29.275  14.287 28.927  13.952 28.575  13.622 28.218  13.295 27.858  12.973 27.493  12.655 27.124  12.341 26.752  12.032 26.376  11.726 25.995  11.426 25.611  11.129 25.224  10.837 24.833  10.550 24.438  10.267 24.039  9.989 23.637  9.716 23.232  9.447 22.823  9.183 22.411  8.924 21.996  8.669 21.578  8.420 21.156  8.175 20.732  7.936 20.304  7.701 19.873  7.471 19.440  7.246 19.004  7.027 18.565  6.813 18.123  6.603 17.679  6.399 17.232  6.200 16.783  6.007 16.331  5.818 15.877  5.635 15.420  5.458 14.961  5.285 14.501  5.118 14.038  4.957 13.573  4.801 13.106  4.650 12.637  4.505 12.166  4.365 11.694  4.231 11.219  4.103 10.744  3.980 10.266  3.862 9.787  3.750 9.307  3.644 8.825  3.544 8.343  3.449 7.858  3.360 7.373  3.276 6.887  3.198 6.400  3.126 5.911  3.060 5.422  2.999 4.932  2.944 4.442  2.895 3.951  2.852 3.459  2.814 2.966  2.782 2.474  2.756 1.981  2.736 1.487  2.721 0.994  2.713 0.500  2.713 -0.500  2.721 -0.994  2.736 -1.487  2.756 -1.981  2.782 -2.474  2.814 -2.966  2.852 -3.459  2.895 -3.951  2.944 -4.442  2.999 -4.932  3.060 -5.422  3.126 -5.911  3.198 -6.400  3.276 -6.887  3.360 -7.373  3.449 -7.858  3.544 -8.343  3.644 -8.825  3.750 -9.307  3.862 -9.787  3.980 -10.266  4.103 -10.744  4.231 -11.219  4.365 -11.694  4.505 -12.166  4.650 -12.637  4.801 -13.106  4.957 -13.573  5.118 -14.038  5.285 -14.501  5.458 -14.961  5.635 -15.420  5.818 -15.877  6.007 -16.331  6.200 -16.783  6.399 -17.232  6.603 -17.679  6.813 -18.123  7.027 -18.565  7.246 -19.004  7.471 -19.440  7.701 -19.873  7.936 -20.304  8.175 -20.732  8.420 -21.156  8.669 -21.578  8.924 -21.996  9.183 -22.411  9.447 -22.823  9.716 -23.232  9.989 -23.637  10.267 -24.039  10.550 -24.438  10.837 -24.833  11.129 -25.224  11.426 -25.611  11.726 -25.995  12.032 -26.376  12.341 -26.752  12.655 -27.124  12.973 -27.493  13.295 -27.858  13.622 -28.218  13.952 -28.575  14.287 -28.927  14.625 -29.275  14.968 -29.619  15.314 -29.959  15.665 -30.294  16.019 -30.625  16.376 -30.952  16.738 -31.274  17.103 -31.591  17.471 -31.904  17.844 -32.213  18.219 -32.516  18.598 -32.815  18.980 -33.110  19.366 -33.399  19.755 -33.684  20.147 -33.964  20.542 -34.239  20.940 -34.509  21.341 -34.774  21.745 -35.034  22.151 -35.289  22.561 -35.539  22.973 -35.784  23.388 -36.023  23.805 -36.258  24.225 -36.487  24.648 -36.711  25.073 -36.930  25.500 -37.143  25.929 -37.351  26.361 -37.554  26.794 -37.751  27.230 -37.943  27.668 -38.130  28.107 -38.311  28.549 -38.486  28.992 -38.656  29.437 -38.820  29.884 -38.979  30.332 -39.132  30.781 -39.280  31.233 -39.422  31.685 -39.558  32.139 -39.689  32.594 -39.814  33.050 -39.933  33.507 -40.047  33.965 -40.155  34.424 -40.257  34.884 -40.353  35.344 -40.444  35.806 -40.528  36.268 -40.607  36.730 -40.680  37.193 -40.748  37.656 -40.809  38.120 -40.865  38.584 -40.915  39.048 -40.959  39.512 -40.997  39.977 -41.030  40.441 -41.056  40.905 -41.077  41.369 -41.092  41.833 -41.101  42.296 -41.104  42.759 -41.101  43.221 -41.092  43.683 -41.078  44.144 -41.058  44.605 -41.032  45.065 -41.000  45.523 -40.962  45.981 -40.919  46.438 -40.869  46.894 -40.814  47.349 -40.753  47.802 -40.687  48.254 -40.614  48.705 -40.536  49.154 -40.452  49.601 -40.362  50.047 -40.267  50.492 -40.166  50.934 -40.059  51.375 -39.947  51.814 -39.829  52.251 -39.705  52.685 -39.576  53.118 -39.442  53.548 -39.301  53.976 -39.155  54.402 -39.004  54.825 -38.847  55.246 -38.685  55.664 -38.517  56.080 -38.344  56.493 -38.166  56.903 -37.982  57.310 -37.793  57.714 -37.598  58.115 -37.399  58.513 -37.194  58.908 -36.984  59.300 -36.768  59.688 -36.548  60.073 -36.322  60.455 -36.092  60.833 -35.856  61.207 -35.616  61.578 -35.370  61.945 -35.120  62.308 -34.864  62.668 -34.604  63.023 -34.339  63.374 -34.069  63.722 -33.794  64.065 -33.515  64.404 -33.231  64.738 -32.943  65.068 -32.650  65.394 -32.352  65.715 -32.050  66.032 -31.743  66.344 -31.432  66.651 -31.117  66.954 -30.798  67.251 -30.474  67.544 -30.146  67.832 -29.813  68.114 -29.477  68.392 -29.136  68.664 -28.792  68.931 -28.443  69.192 -28.090  69.448 -27.734  69.699 -27.374  69.944 -27.009  70.183 -26.641  70.416 -26.270  70.644 -25.894  70.866 -25.515  71.081 -25.133  71.291 -24.746  71.495 -24.357  71.692 -23.963  71.883 -23.567  72.067 -23.166  72.245 -22.763  72.417 -22.356  72.582 -21.946  72.740 -21.533  72.891 -21.116  73.035 -20.696  73.172 -20.273  73.301 -19.847  73.424 -19.418  73.539 -18.985  73.646 -18.550  73.746 -18.111  73.837 -17.669  73.921 -17.224  73.996 -16.776  74.063 -16.325  74.122 -15.871  74.171 -15.413  74.212 -14.952  74.243 -14.488  74.265 -14.021  74.277 -13.550  74.279 -13.075  74.270 -12.597  74.250 -12.115  74.219 -11.629  74.176 -11.138  74.121 -10.642  74.052 -10.142  73.969 -9.636  73.871 -9.124  73.757 -8.605  73.625 -8.078  73.473 -7.543  73.300 -6.997  73.101 -6.439  72.874 -5.866  72.612 -5.273  72.305 -4.655  71.940 -4.002  71.489 -3.295  70.888 -2.486  69.892 -1.360   /
\color{Blue} 
\multiput {\large$\cdot$} at
69.957 0.959  70.686 1.705  71.133 2.222  71.475 2.665  71.759 3.068  72.004 3.443  72.220 3.800  72.415 4.142  72.592 4.473  72.753 4.793  72.903 5.106  73.040 5.413  73.168 5.713  73.288 6.008  73.399 6.299  73.503 6.586  73.600 6.869  73.691 7.148  73.776 7.425  73.855 7.699  73.929 7.971  73.999 8.240  74.063 8.507  74.124 8.772  74.180 9.035  74.231 9.296  74.279 9.556  74.323 9.814  74.364 10.070  74.401 10.325  74.434 10.579  74.464 10.831  74.491 11.082  74.515 11.332  74.535 11.581  74.553 11.828  74.568 12.074  74.579 12.320  74.588 12.564  74.595 12.807  74.598 13.050  74.599 13.291  74.597 13.531  74.593 13.771  74.586 14.009  74.576 14.247  74.565 14.484  74.550 14.719  74.534 14.954  74.515 15.189  74.494 15.422  74.470 15.654  74.444 15.886  74.416 16.117  74.386 16.347  74.354 16.576  74.319 16.805  74.283 17.032  74.244 17.259  74.203 17.485  74.160 17.710  74.115 17.935  74.069 18.158  74.020 18.381  73.969 18.603  73.916 18.825  73.861 19.045  73.804 19.265  73.746 19.484  73.685 19.702  73.623 19.919  73.559 20.136  73.493 20.351  73.425 20.566  73.355 20.780  73.284 20.994  73.210 21.206  73.135 21.418  73.058 21.629  72.980 21.839  72.899 22.048  72.817 22.256  72.734 22.464  72.648 22.670  72.561 22.876  72.472 23.081  72.382 23.285  72.290 23.489  72.196 23.691  72.101 23.893  72.004 24.094  71.905 24.293  71.805 24.492  71.703 24.690  71.600 24.888  71.495 25.084  71.389 25.279  71.281 25.474  71.172 25.668  71.061 25.860  70.948 26.052  70.834 26.243  70.719 26.433  70.602 26.622  70.483 26.810  70.363 26.997  70.242 27.183  70.119 27.368  69.995 27.552  69.870 27.736  69.743 27.918  69.614 28.099  69.484 28.280  69.353 28.459  69.221 28.637  69.087 28.815  68.952 28.991  68.815 29.166  68.677 29.340  68.538 29.514  68.397 29.686  68.255 29.857  68.112 30.027  67.968 30.196  67.822 30.364  67.675 30.531  67.527 30.697  67.377 30.862  67.226 31.026  67.074 31.189  66.921 31.350  66.767 31.511  66.611 31.670  66.454 31.829  66.296 31.986  66.137 32.142  65.977 32.297  65.815 32.451  65.652 32.604  65.488 32.755  65.323 32.906  65.157 33.055  64.990 33.203  64.822 33.350  64.652 33.496  64.482 33.641  64.310 33.784  64.138 33.926  63.964 34.068  63.789 34.208  63.613 34.346  63.436 34.484  63.258 34.620  63.079 34.755  62.899 34.889  62.718 35.022  62.536 35.154  62.354 35.284  62.170 35.413  61.985 35.541  61.799 35.667  61.612 35.793  61.424 35.917  61.236 36.040  61.046 36.161  60.856 36.282  60.664 36.401  60.472 36.518  60.279 36.635  60.085 36.750  59.890 36.864  59.694 36.977  59.498 37.088  59.300 37.198  59.102 37.307  58.903 37.414  58.703 37.520  58.502 37.625  58.301 37.729  58.098 37.831  57.895 37.932  57.692 38.031  57.487 38.129  57.282 38.226  57.076 38.322  56.869 38.416  56.661 38.508  56.453 38.600  56.244 38.690  56.034 38.779  55.824 38.866  55.613 38.952  55.401 39.037  55.189 39.120  54.976 39.202  54.763 39.282  54.548 39.361  54.333 39.439  54.118 39.515  53.902 39.590  53.685 39.664  53.468 39.736  53.250 39.807  53.032 39.876  52.813 39.944  52.594 40.011  52.374 40.076  52.153 40.139  51.932 40.202  51.710 40.263  51.488 40.322  51.266 40.380  51.043 40.437  50.819 40.492  50.596 40.545  50.371 40.598  50.146 40.649  49.921 40.698  49.695 40.746  49.469 40.792  49.243 40.838  49.016 40.881  48.789 40.923  48.561 40.964  48.333 41.003  48.105 41.041  47.876 41.078  47.647 41.113  47.418 41.146  47.188 41.178  46.958 41.208  46.728 41.238  46.497 41.265  46.266 41.291  46.035 41.316  45.804 41.339  45.572 41.361  45.340 41.381  45.108 41.400  44.876 41.417  44.643 41.433  44.411 41.448  44.178 41.460  43.945 41.472  43.711 41.482  43.478 41.490  43.244 41.497  43.010 41.503  42.777 41.507  42.543 41.509  42.308 41.510  42.074 41.510  41.840 41.508  41.605 41.505  41.371 41.500  41.136 41.494  40.902 41.486  40.667 41.477  40.432 41.466  40.198 41.454  39.963 41.440  39.728 41.425  39.493 41.408  39.258 41.390  39.024 41.370  38.789 41.349  38.554 41.326  38.319 41.302  38.085 41.277  37.850 41.250  37.616 41.221  37.381 41.191  37.147 41.160  36.913 41.127  36.679 41.093  36.445 41.057  36.211 41.019  35.977 40.981  35.744 40.940  35.510 40.898  35.277 40.855  35.044 40.811  34.811 40.764  34.578 40.717  34.346 40.668  34.114 40.617  33.881 40.565  33.650 40.512  33.418 40.457  33.187 40.401  32.956 40.343  32.725 40.284  32.494 40.223  32.264 40.161  32.034 40.097  31.804 40.032  31.575 39.966  31.346 39.898  31.117 39.829  30.889 39.758  30.661 39.686  30.433 39.612  30.206 39.537  29.979 39.461  29.753 39.383  29.527 39.304  29.301 39.223  29.076 39.141  28.851 39.058  28.626 38.973  28.402 38.886  28.179 38.799  27.956 38.710  27.733 38.619  27.511 38.527  27.289 38.434  27.068 38.339  26.847 38.243  26.627 38.146  26.407 38.047  26.188 37.947  25.970 37.845  25.752 37.743  25.534 37.638  25.317 37.533  25.101 37.426  24.885 37.317  24.670 37.208  24.455 37.097  24.241 36.984  24.028 36.871  23.815 36.756  23.603 36.640  23.391 36.522  23.180 36.403  22.970 36.283  22.761 36.161  22.552 36.038  22.343 35.914  22.136 35.789  21.929 35.662  21.723 35.534  21.517 35.404  21.313 35.274  21.109 35.142  20.906 35.009  20.703 34.874  20.501 34.739  20.300 34.602  20.100 34.464  19.900 34.324  19.702 34.184  19.504 34.042  19.307 33.899  19.110 33.754  18.915 33.609  18.720 33.462  18.526 33.314  18.333 33.165  18.141 33.014  17.950 32.863  17.759 32.710  17.570 32.556  17.381 32.401  17.193 32.244  17.006 32.087  16.820 31.928  16.635 31.769  16.451 31.608  16.267 31.446  16.085 31.282  15.904 31.118  15.723 30.952  15.543 30.786  15.365 30.618  15.187 30.449  15.010 30.279  14.834 30.108  14.660 29.936  14.486 29.763  14.313 29.589  14.141 29.413  13.970 29.237  13.801 29.060  13.632 28.881  13.464 28.702  13.297 28.521  13.132 28.339  12.967 28.157  12.803 27.973  12.641 27.788  12.479 27.603  12.319 27.416  12.160 27.228  12.001 27.039  11.844 26.850  11.688 26.659  11.533 26.468  11.379 26.275  11.226 26.081  11.075 25.887  10.924 25.692  10.775 25.495  10.627 25.298  10.479 25.100  10.333 24.901  10.189 24.701  10.045 24.500  9.903 24.298  9.761 24.095  9.621 23.892  9.482 23.688  9.344 23.482  9.208 23.276  9.072 23.069  8.938 22.861  8.805 22.653  8.673 22.443  8.543 22.233  8.413 22.022  8.285 21.810  8.158 21.598  8.033 21.384  7.908 21.170  7.785 20.955  7.663 20.739  7.543 20.523  7.423 20.306  7.305 20.088  7.188 19.869  7.073 19.650  6.958 19.429  6.845 19.209  6.734 18.987  6.623 18.765  6.514 18.542  6.406 18.318  6.300 18.094  6.194 17.869  6.090 17.644  5.988 17.417  5.886 17.191  5.787 16.963  5.688 16.735  5.591 16.506  5.495 16.277  5.400 16.047  5.307 15.817  5.215 15.586  5.124 15.354  5.035 15.122  4.947 14.889  4.860 14.656  4.775 14.422  4.692 14.188  4.609 13.953  4.528 13.718  4.448 13.482  4.370 13.246  4.293 13.009  4.218 12.772  4.144 12.534  4.071 12.296  3.999 12.057  3.929 11.818  3.861 11.579  3.794 11.339  3.728 11.099  3.664 10.858  3.601 10.617  3.539 10.375  3.479 10.133  3.420 9.891  3.363 9.649  3.307 9.406  3.253 9.162  3.200 8.919  3.148 8.675  3.098 8.430  3.049 8.186  3.002 7.941  2.956 7.696  2.912 7.450  2.869 7.205  2.827 6.959  2.787 6.713  2.749 6.466  2.712 6.219  2.676 5.972  2.642 5.725  2.609 5.478  2.577 5.230  2.547 4.983  2.519 4.735  2.492 4.487  2.466 4.238  2.442 3.990  2.419 3.741  2.398 3.493  2.378 3.244  2.360 2.995  2.343 2.746  2.328 2.497  2.314 2.247  2.302 1.998  2.291 1.749  2.281 1.499  2.273 1.250  2.266 1.000  2.261 0.751  2.258 0.501  2.255 0.251  2.255 -0.251  2.258 -0.501  2.261 -0.751  2.266 -1.000  2.273 -1.250  2.281 -1.499  2.291 -1.749  2.302 -1.998  2.314 -2.247  2.328 -2.497  2.343 -2.746  2.360 -2.995  2.378 -3.244  2.398 -3.493  2.419 -3.741  2.442 -3.990  2.466 -4.238  2.492 -4.487  2.519 -4.735  2.547 -4.983  2.577 -5.230  2.609 -5.478  2.642 -5.725  2.676 -5.972  2.712 -6.219  2.749 -6.466  2.787 -6.713  2.827 -6.959  2.869 -7.205  2.912 -7.450  2.956 -7.696  3.002 -7.941  3.049 -8.186  3.098 -8.430  3.148 -8.675  3.200 -8.919  3.253 -9.162  3.307 -9.406  3.363 -9.649  3.420 -9.891  3.479 -10.133  3.539 -10.375  3.601 -10.617  3.664 -10.858  3.728 -11.099  3.794 -11.339  3.861 -11.579  3.929 -11.818  3.999 -12.057  4.071 -12.296  4.144 -12.534  4.218 -12.772  4.293 -13.009  4.370 -13.246  4.448 -13.482  4.528 -13.718  4.609 -13.953  4.692 -14.188  4.775 -14.422  4.860 -14.656  4.947 -14.889  5.035 -15.122  5.124 -15.354  5.215 -15.586  5.307 -15.817  5.400 -16.047  5.495 -16.277  5.591 -16.506  5.688 -16.735  5.787 -16.963  5.886 -17.191  5.988 -17.417  6.090 -17.644  6.194 -17.869  6.300 -18.094  6.406 -18.318  6.514 -18.542  6.623 -18.765  6.734 -18.987  6.845 -19.209  6.958 -19.429  7.073 -19.650  7.188 -19.869  7.305 -20.088  7.423 -20.306  7.543 -20.523  7.663 -20.739  7.785 -20.955  7.908 -21.170  8.033 -21.384  8.158 -21.598  8.285 -21.810  8.413 -22.022  8.543 -22.233  8.673 -22.443  8.805 -22.653  8.938 -22.861  9.072 -23.069  9.208 -23.276  9.344 -23.482  9.482 -23.688  9.621 -23.892  9.761 -24.095  9.903 -24.298  10.045 -24.500  10.189 -24.701  10.333 -24.901  10.479 -25.100  10.627 -25.298  10.775 -25.495  10.924 -25.692  11.075 -25.887  11.226 -26.081  11.379 -26.275  11.533 -26.468  11.688 -26.659  11.844 -26.850  12.001 -27.039  12.160 -27.228  12.319 -27.416  12.479 -27.603  12.641 -27.788  12.803 -27.973  12.967 -28.157  13.132 -28.339  13.297 -28.521  13.464 -28.702  13.632 -28.881  13.801 -29.060  13.970 -29.237  14.141 -29.413  14.313 -29.589  14.486 -29.763  14.660 -29.936  14.834 -30.108  15.010 -30.279  15.187 -30.449  15.365 -30.618  15.543 -30.786  15.723 -30.952  15.904 -31.118  16.085 -31.282  16.267 -31.446  16.451 -31.608  16.635 -31.769  16.820 -31.928  17.006 -32.087  17.193 -32.244  17.381 -32.401  17.570 -32.556  17.759 -32.710  17.950 -32.863  18.141 -33.014  18.333 -33.165  18.526 -33.314  18.720 -33.462  18.915 -33.609  19.110 -33.754  19.307 -33.899  19.504 -34.042  19.702 -34.184  19.900 -34.324  20.100 -34.464  20.300 -34.602  20.501 -34.739  20.703 -34.874  20.906 -35.009  21.109 -35.142  21.313 -35.274  21.517 -35.404  21.723 -35.534  21.929 -35.662  22.136 -35.789  22.343 -35.914  22.552 -36.038  22.761 -36.161  22.970 -36.283  23.180 -36.403  23.391 -36.522  23.603 -36.640  23.815 -36.756  24.028 -36.871  24.241 -36.984  24.455 -37.097  24.670 -37.208  24.885 -37.317  25.101 -37.426  25.317 -37.533  25.534 -37.638  25.752 -37.743  25.970 -37.845  26.188 -37.947  26.407 -38.047  26.627 -38.146  26.847 -38.243  27.068 -38.339  27.289 -38.434  27.511 -38.527  27.733 -38.619  27.956 -38.710  28.179 -38.799  28.402 -38.886  28.626 -38.973  28.851 -39.058  29.076 -39.141  29.301 -39.223  29.527 -39.304  29.753 -39.383  29.979 -39.461  30.206 -39.537  30.433 -39.612  30.661 -39.686  30.889 -39.758  31.117 -39.829  31.346 -39.898  31.575 -39.966  31.804 -40.032  32.034 -40.097  32.264 -40.161  32.494 -40.223  32.725 -40.284  32.956 -40.343  33.187 -40.401  33.418 -40.457  33.650 -40.512  33.881 -40.565  34.114 -40.617  34.346 -40.668  34.578 -40.717  34.811 -40.764  35.044 -40.811  35.277 -40.855  35.510 -40.898  35.744 -40.940  35.977 -40.981  36.211 -41.019  36.445 -41.057  36.679 -41.093  36.913 -41.127  37.147 -41.160  37.381 -41.191  37.616 -41.221  37.850 -41.250  38.085 -41.277  38.319 -41.302  38.554 -41.326  38.789 -41.349  39.024 -41.370  39.258 -41.390  39.493 -41.408  39.728 -41.425  39.963 -41.440  40.198 -41.454  40.432 -41.466  40.667 -41.477  40.902 -41.486  41.136 -41.494  41.371 -41.500  41.605 -41.505  41.840 -41.508  42.074 -41.510  42.308 -41.510  42.543 -41.509  42.777 -41.507  43.010 -41.503  43.244 -41.497  43.478 -41.490  43.711 -41.482  43.945 -41.472  44.178 -41.460  44.411 -41.448  44.643 -41.433  44.876 -41.417  45.108 -41.400  45.340 -41.381  45.572 -41.361  45.804 -41.339  46.035 -41.316  46.266 -41.291  46.497 -41.265  46.728 -41.238  46.958 -41.208  47.188 -41.178  47.418 -41.146  47.647 -41.113  47.876 -41.078  48.105 -41.041  48.333 -41.003  48.561 -40.964  48.789 -40.923  49.016 -40.881  49.243 -40.838  49.469 -40.792  49.695 -40.746  49.921 -40.698  50.146 -40.649  50.371 -40.598  50.596 -40.545  50.819 -40.492  51.043 -40.437  51.266 -40.380  51.488 -40.322  51.710 -40.263  51.932 -40.202  52.153 -40.139  52.374 -40.076  52.594 -40.011  52.813 -39.944  53.032 -39.876  53.250 -39.807  53.468 -39.736  53.685 -39.664  53.902 -39.590  54.118 -39.515  54.333 -39.439  54.548 -39.361  54.763 -39.282  54.976 -39.202  55.189 -39.120  55.401 -39.037  55.613 -38.952  55.824 -38.866  56.034 -38.779  56.244 -38.690  56.453 -38.600  56.661 -38.508  56.869 -38.416  57.076 -38.322  57.282 -38.226  57.487 -38.129  57.692 -38.031  57.895 -37.932  58.098 -37.831  58.301 -37.729  58.502 -37.625  58.703 -37.520  58.903 -37.414  59.102 -37.307  59.300 -37.198  59.498 -37.088  59.694 -36.977  59.890 -36.864  60.085 -36.750  60.279 -36.635  60.472 -36.518  60.664 -36.401  60.856 -36.282  61.046 -36.161  61.236 -36.040  61.424 -35.917  61.612 -35.793  61.799 -35.667  61.985 -35.541  62.170 -35.413  62.354 -35.284  62.536 -35.154  62.718 -35.022  62.899 -34.889  63.079 -34.755  63.258 -34.620  63.436 -34.484  63.613 -34.346  63.789 -34.208  63.964 -34.068  64.138 -33.926  64.310 -33.784  64.482 -33.641  64.652 -33.496  64.822 -33.350  64.990 -33.203  65.157 -33.055  65.323 -32.906  65.488 -32.755  65.652 -32.604  65.815 -32.451  65.977 -32.297  66.137 -32.142  66.296 -31.986  66.454 -31.829  66.611 -31.670  66.767 -31.511  66.921 -31.350  67.074 -31.189  67.226 -31.026  67.377 -30.862  67.527 -30.697  67.675 -30.531  67.822 -30.364  67.968 -30.196  68.112 -30.027  68.255 -29.857  68.397 -29.686  68.538 -29.514  68.677 -29.340  68.815 -29.166  68.952 -28.991  69.087 -28.815  69.221 -28.637  69.353 -28.459  69.484 -28.280  69.614 -28.099  69.743 -27.918  69.870 -27.736  69.995 -27.552  70.119 -27.368  70.242 -27.183  70.363 -26.997  70.483 -26.810  70.602 -26.622  70.719 -26.433  70.834 -26.243  70.948 -26.052  71.061 -25.860  71.172 -25.668  71.281 -25.474  71.389 -25.279  71.495 -25.084  71.600 -24.888  71.703 -24.690  71.805 -24.492  71.905 -24.293  72.004 -24.094  72.101 -23.893  72.196 -23.691  72.290 -23.489  72.382 -23.285  72.472 -23.081  72.561 -22.876  72.648 -22.670  72.734 -22.464  72.817 -22.256  72.899 -22.048  72.980 -21.839  73.058 -21.629  73.135 -21.418  73.210 -21.206  73.284 -20.994  73.355 -20.780  73.425 -20.566  73.493 -20.351  73.559 -20.136  73.623 -19.919  73.685 -19.702  73.746 -19.484  73.804 -19.265  73.861 -19.045  73.916 -18.825  73.969 -18.603  74.020 -18.381  74.069 -18.158  74.115 -17.935  74.160 -17.710  74.203 -17.485  74.244 -17.259  74.283 -17.032  74.319 -16.805  74.354 -16.576  74.386 -16.347  74.416 -16.117  74.444 -15.886  74.470 -15.654  74.494 -15.422  74.515 -15.189  74.534 -14.954  74.550 -14.719  74.565 -14.484  74.576 -14.247  74.586 -14.009  74.593 -13.771  74.597 -13.531  74.599 -13.291  74.598 -13.050  74.595 -12.807  74.588 -12.564  74.579 -12.320  74.568 -12.074  74.553 -11.828  74.535 -11.581  74.515 -11.332  74.491 -11.082  74.464 -10.831  74.434 -10.579  74.401 -10.325  74.364 -10.070  74.323 -9.814  74.279 -9.556  74.231 -9.296  74.180 -9.035  74.124 -8.772  74.063 -8.507  73.999 -8.240  73.929 -7.971  73.855 -7.699  73.776 -7.425  73.691 -7.148  73.600 -6.869  73.503 -6.586  73.399 -6.299  73.288 -6.008  73.168 -5.713  73.040 -5.413  72.903 -5.106  72.753 -4.793  72.592 -4.473  72.415 -4.142  72.220 -3.800  72.004 -3.443  71.759 -3.068  71.475 -2.665  71.133 -2.222  70.686 -1.705  69.957 -0.959 /

\setplotsymbol ({\Large$\cdot$})
\color{black}
\plot 70.00 0.00  70.62 0.66  71.21 1.37  71.78 2.15  72.31 2.98  72.81 3.86  73.26 4.80  73.67 5.78  74.03 6.82  74.34 7.90  74.60 9.02  74.79 10.18  74.92 11.38  74.99 12.61  74.99 13.86  74.92 15.14  74.78 16.44  74.56 17.76  74.27 19.09  73.90 20.42  73.45 21.76  72.92 23.09  72.31 24.41  71.62 25.72  70.85 27.01  70.00 28.28  69.07 29.53  68.07 30.74  66.98 31.91  65.83 33.04  64.60 34.12  63.30 35.15  61.93 36.12  60.49 37.03  59.00 37.88  57.44 38.66  55.83 39.36  54.17 39.99  52.46 40.54  50.71 41.00  48.91 41.38  47.09 41.67  45.24 41.87  43.36 41.98  41.46 41.99  39.56 41.90  37.64 41.72  35.72 41.44  33.80 41.06  31.89 40.58  30.00 40.00  28.13 39.32  26.28 38.55  24.46 37.67  22.68 36.71  20.94 35.65  19.24 34.49  17.60 33.25  16.02 31.92  14.50 30.51  13.04 29.02  11.66 27.45  10.35 25.81  9.12 24.10  7.98 22.33  6.92 20.50  5.95 18.61  5.08 16.67  4.30 14.69  3.62 12.66  3.04 10.60  2.57 8.52  2.20 6.41  1.93 4.28  1.77 2.14  1.72 -0.00  1.77 -2.14  1.93 -4.28  2.20 -6.41  2.57 -8.52  3.04 -10.60  3.62 -12.66  4.30 -14.69  5.08 -16.67  5.95 -18.61  6.92 -20.50  7.98 -22.33  9.12 -24.10  10.35 -25.81  11.66 -27.45  13.04 -29.02  14.50 -30.51  16.02 -31.92  17.60 -33.25  19.24 -34.49  20.94 -35.65  22.68 -36.71  24.46 -37.67  26.28 -38.55  28.13 -39.32  30.00 -40.00  31.89 -40.58  33.80 -41.06  35.72 -41.44  37.64 -41.72  39.56 -41.90  41.46 -41.99  43.36 -41.98  45.24 -41.87  47.09 -41.67  48.91 -41.38  50.71 -41.00  52.46 -40.54  54.17 -39.99  55.83 -39.36  57.44 -38.66  59.00 -37.88  60.49 -37.03  61.93 -36.12  63.30 -35.15  64.60 -34.12  65.83 -33.04  66.98 -31.91  68.07 -30.74  69.07 -29.53  70.00 -28.28  70.85 -27.01  71.62 -25.72  72.31 -24.41  72.92 -23.09  73.45 -21.76  73.90 -20.42  74.27 -19.09  74.56 -17.76  74.78 -16.44  74.92 -15.14  74.99 -13.86  74.99 -12.61  74.92 -11.38  74.79 -10.18  74.60 -9.02  74.34 -7.90  74.03 -6.82  73.67 -5.78  73.26 -4.80  72.81 -3.86  72.31 -2.98  71.78 -2.15  71.21 -1.37  70.62 -0.66 70 0 /

\color{black}
\normalcolor

\endpicture
\end{center}
\caption{Left panel: Approximate zeros  computed from equation \Ref{eqn31}.
The $a''_+$ are denoted by $\circ$, the $a''_-$ are denoted by
$\times$, and exact zeros by $\bullet$.  The data are for $N=16$ and $n=32$.
Notice that the $a''_-$ are located on the other branch of the lima\c{c}on
while the $a''_+$ are approximate zeros on the first branch.
Right panel:  The zeros $a''_+$ for 
$n\in\{16,32,64,128,256,512,1024\}$.  For larger values of $n$ the zeros approach
the lima\c{c}on which is plotted around the zeros.}
\label{figure5}
\end{figure}

In the left panel in Figure \ref{figure5} the $a''_\pm$ are plotted for $n=16$ and
$n=32$.  The bullets are exact locations of partition function zeros, and the estimates
$a''_\pm$ are shown by open circles (for $a''_+$) and crosses (for $a''_-)$.
This shows that the choice of the $+$-sign in equation \Ref{eqn31} gives approximations
to the partition function zeros, while the $a''_-$ are points located close to the
other branch $\lmI$ of the lima\c{c}on. Notice in particular that equation \Ref{eqn31} does not
produce approximation to all the zeros -- there are two zeros for each value of $n$
near the negative real axis which are not approximated.  In addition, there are two
extra approximations for each value of $n$ near the positive real axis.  The remaining
zeros are approximated well.

In the right panel of Figure \ref{figure5}, the approximate zeros $a''_+$ are plotted for
$n\in\{16,32,64,128,256,512,1024 \}$.  With increasing $n$ these approach the lima\c{c}on.
Putting $k=\lfloor \rho n \rfloor$ in equation \Ref{eqn31} and then taking $n\to\infty$ gives
\begin{equation}
a''(\rho)
 = \lim_{n\to\infty} a''_+ 
= 2e^{2\rho\pi\Imi} + 2e^{\rho\pi\Imi}
\sqrt{ e^{2\rho\pi\Imi} - 1}
\label{eqn33A} 
\end{equation}
which is an alternative parametrization of the lima\c{c}on given by equation \Ref{eqn4}.

\begin{figure}[t]
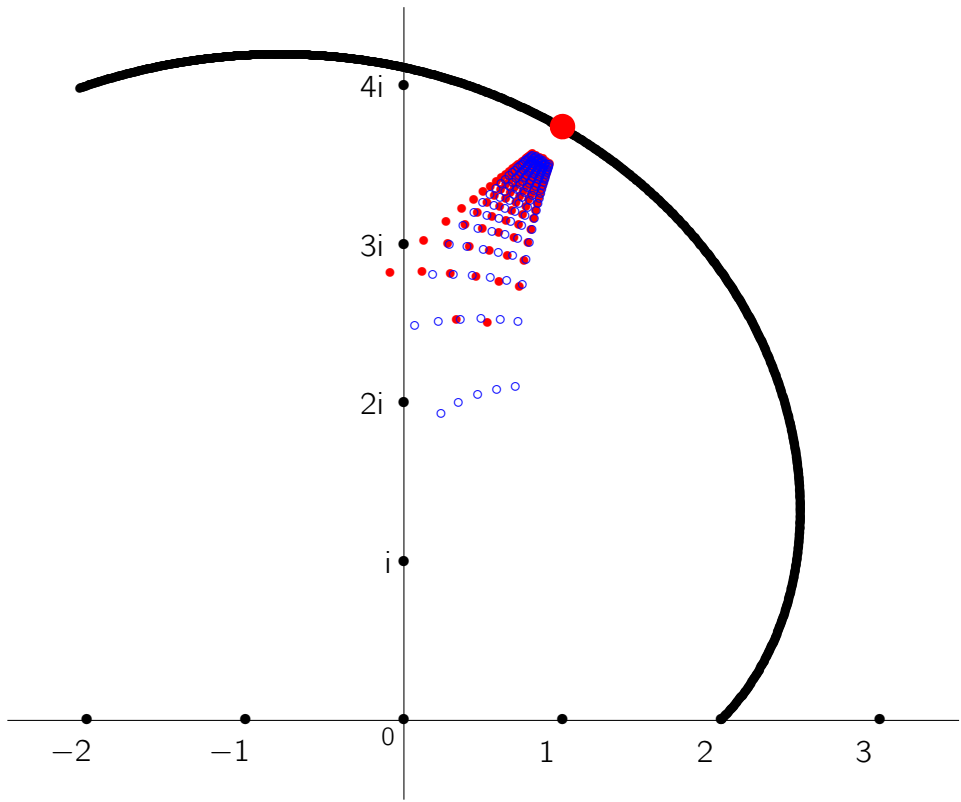

\begin{center}
\beginpicture
\color{black}
\setcoordinatesystem units <6pt,6pt>
\setplotarea x from -25 to 35, y from -5 to 45
\axes{-25}{-5}{0}{0}{35}{45}{-30}
\put {\footnotesize$0$} at -1 -1
\put {\footnotesize$\bullet$} at -20 0
\put {$-2$} at -21 -2
\put {\footnotesize$\bullet$} at -10 0
\put {$-1$} at -11 -2
\put {\footnotesize$\bullet$} at 10 0
\put {$1$} at 9 -2
\put {\footnotesize$\bullet$} at 20 0
\put {$2$} at 19 -2
\put {\footnotesize$\bullet$} at 30 0
\put {$3$} at 29 -2
\put {\footnotesize$\bullet$} at 0 10
\put {$\Imi$} at -1 10
\put {\footnotesize$\bullet$} at 0 20
\put {$2\Imi$} at -2 20
\put {\footnotesize$\bullet$} at 0 30
\put {$3\Imi$} at -2 30
\put {\footnotesize$\bullet$} at 0 40
\put {$4\Imi$} at -2 40

\color{Red}
\multiput {\scriptsize$\bullet$} at 3.322 25.251  5.258 25.082   -0.871 28.213  1.141 28.268  2.941 28.173  4.552 27.969  5.996 27.688  7.292 27.353  1.256 30.218  2.757 30.078  4.129 29.869  5.384 29.609  6.533 29.311  7.587 28.985  2.658 31.430  3.851 31.227  4.957 30.987  5.983 30.719  6.937 30.430  7.824 30.123  3.655 32.252  4.642 32.031  5.568 31.787  6.436 31.527  7.251 31.253  8.017 30.968  4.402 32.851  5.244 32.628  6.039 32.392  6.791 32.144  7.503 31.887  8.177 31.623  4.984 33.309  5.717 33.093  6.414 32.866  7.078 32.632  7.709 32.391  8.310 32.146  5.451 33.673  6.100 33.465  6.720 33.250  7.314 33.029  7.881 32.804  8.424 32.574  5.834 33.970  6.417 33.771  6.975 33.567  7.512 33.359  8.027 33.147  8.522 32.933  6.155 34.217  6.683 34.028  7.191 33.835  7.681 33.638  8.153 33.439  8.608 33.237  6.428 34.427  6.911 34.247  7.377 34.063  7.827 33.877  8.263 33.689  8.683 33.499  6.662 34.608  7.107 34.436  7.538 34.261  7.955 34.085  8.359 33.907  8.750 33.727  6.867 34.765  7.279 34.601  7.679 34.435  8.067 34.267  8.444 34.098  8.810 33.928  7.047 34.903  7.431 34.746  7.804 34.588  8.167 34.428  8.520 34.268  8.864 34.106  7.206 35.025  7.565 34.875  7.916 34.724  8.257 34.572  8.589 34.419  8.912 34.265  7.348 35.135  7.686 34.991  8.016 34.847  8.337 34.701  8.651 34.555  8.957 34.408  7.476 35.233  7.794 35.095  8.106 34.957  8.410 34.818  8.707 34.678  8.997 34.537  7.591 35.322  7.893 35.190  8.188 35.057  8.476 34.924  8.758 34.789  9.034 34.655  7.696 35.403  7.982 35.276  8.263 35.148  8.537 35.020  8.806 34.891  9.069 34.762  7.791 35.478  8.064 35.355  8.331 35.232  8.593 35.109  8.849 34.985  9.100 34.860  7.879 35.546  8.139 35.428  8.394 35.309  8.644 35.190  8.889 35.071  9.130 34.951  7.960 35.608  8.208 35.494  8.452 35.380  8.692 35.265  8.926 35.150  9.157 35.035  8.034 35.667  8.272 35.556  8.506 35.446  8.736 35.335  8.961 35.224  9.183 35.112  8.103 35.720  8.332 35.614  8.556 35.507  8.777 35.400     /

\color{Blue}
\multiput {\scriptsize$\circ$} at 2.342 19.349  3.447 20.027  4.649 20.531  5.863 20.860  7.039 21.037  0.685 24.857  2.173 25.135  3.572 25.277  4.875 25.310  6.084 25.254  7.201 25.129  1.810 28.097  3.113 28.092  4.324 28.015  5.449 27.882  6.494 27.704  7.463 27.490  2.880 29.970  3.977 29.849  5.004 29.689  5.964 29.496  6.862 29.278  7.703 29.040  3.742 31.181  4.677 31.013  5.557 30.821  6.387 30.609  7.169 30.382  7.907 30.141  4.427 32.028  5.237 31.842  6.005 31.641  6.732 31.426  7.423 31.201  8.078 30.968  4.979 32.656  5.691 32.465  6.370 32.264  7.017 32.053  7.634 31.836  8.223 31.612  5.431 33.140  6.066 32.951  6.673 32.754  7.255 32.551  7.812 32.343  8.346 32.130  5.807 33.525  6.379 33.341  6.928 33.151  7.457 32.957  7.964 32.758  8.453 32.556  6.124 33.840  6.645 33.662  7.146 33.480  7.629 33.294  8.096 33.105  8.545 32.914  6.396 34.103  6.873 33.931  7.334 33.757  7.779 33.579  8.210 33.400  8.627 33.218  6.630 34.325  7.071 34.161  7.497 33.993  7.910 33.824  8.311 33.653  8.699 33.480  6.835 34.516  7.244 34.358  7.641 34.198  8.026 34.036  8.400 33.873  8.763 33.708  7.016 34.682  7.397 34.531  7.768 34.377  8.129 34.223  8.479 34.067  8.821 33.910  7.177 34.828  7.534 34.683  7.882 34.536  8.221 34.388  8.551 34.238  8.873 34.088  7.320 34.957  7.656 34.817  7.984 34.676  8.304 34.534  8.616 34.391  8.920 34.248  7.449 35.072  7.766 34.938  8.076 34.803  8.379 34.666  8.674 34.529  8.963 34.391  7.566 35.176  7.866 35.047  8.160 34.916  8.447 34.785  8.728 34.654  9.003 34.521  7.672 35.270  7.957 35.145  8.236 35.019  8.510 34.893  8.777 34.767  9.039 34.639  7.769 35.355  8.040 35.234  8.306 35.113  8.567 34.992  8.822 34.870  9.073 34.747  7.857 35.432  8.116 35.316  8.371 35.199  8.620 35.082  8.864 34.964  9.104 34.846  7.939 35.504  8.187 35.391  8.430 35.278  8.669 35.165  8.903 35.052  9.133 34.938  8.015 35.569  8.252 35.460  8.485 35.351  8.714 35.242  8.939 35.132  9.160 35.022  8.084 35.630  8.312 35.524  8.536 35.419  8.756 35.313      /

\color{black}
\setplotsymbol ({\scriptsize$\bullet$})
\plot 20.000 0.000  20.201 0.205  20.400 0.417  20.597 0.634  20.791 0.858  20.983 1.088  21.172 1.323  21.358 1.565  21.541 1.813  21.720 2.066  21.897 2.325  22.070 2.590  22.239 2.861  22.404 3.137  22.566 3.419  22.724 3.707  22.877 4.000  23.026 4.299  23.171 4.602  23.311 4.911  23.447 5.226  23.577 5.545  23.703 5.869  23.824 6.199  23.939 6.533  24.050 6.872  24.154 7.215  24.253 7.564  24.347 7.916  24.435 8.273  24.517 8.635  24.592 9.000  24.662 9.369  24.725 9.743  24.783 10.120  24.833 10.501  24.877 10.886  24.915 11.274  24.946 11.665  24.970 12.060  24.987 12.457  24.997 12.858  25.000 13.261  24.996 13.668  24.984 14.076  24.966 14.487  24.939 14.901  24.906 15.316  24.865 15.734  24.816 16.153  24.759 16.574  24.695 16.997  24.623 17.421  24.543 17.847  24.456 18.273  24.360 18.701  24.256 19.129  24.144 19.558  24.025 19.987  23.897 20.417  23.761 20.847  23.616 21.276  23.464 21.706  23.303 22.136  23.134 22.565  22.957 22.993  22.771 23.420  22.578 23.847  22.375 24.272  22.165 24.697  21.946 25.119  21.719 25.540  21.484 25.959  21.240 26.377  20.988 26.792  20.728 27.205  20.459 27.615  20.183 28.023  19.898 28.428  19.605 28.830  19.303 29.229  18.994 29.625  18.677 30.017  18.351 30.406  18.018 30.790  17.676 31.171  17.327 31.548  16.970 31.921  16.605 32.289  16.233 32.652  15.852 33.011  15.465 33.365  15.069 33.714  14.667 34.058  14.257 34.396  13.840 34.729  13.415 35.056  12.984 35.377  12.545 35.692  12.100 36.001  11.647 36.304  11.188 36.601  10.723 36.891  10.251 37.174  9.772 37.451  9.287 37.720  8.796 37.983  8.299 38.238  7.796 38.486  7.288 38.727  6.773 38.959  6.253 39.185  5.728 39.402  5.197 39.611  4.661 39.812  4.120 40.005  3.574 40.190  3.024 40.366  2.468 40.534  1.909 40.693  1.345 40.844  0.777 40.986  0.205 41.118  -0.371 41.242  -0.951 41.357  -1.534 41.462  -2.121 41.558  -2.711 41.645  -3.303 41.723  -3.899 41.791  -4.498 41.849  -5.098 41.898  -5.702 41.937  -6.307 41.966  -6.915 41.985  -7.524 41.995  -8.135 41.995  -8.748 41.984  -9.362 41.964  -9.977 41.933  -10.593 41.893  -11.210 41.842  -11.827 41.781  -12.445 41.710  -13.063 41.629  -13.682 41.538  -14.300 41.436  -14.917 41.324  -15.534 41.202  -16.151 41.069  -16.767 40.926  -17.381 40.773  -17.994 40.610  -18.606 40.436  -19.217 40.252  -19.825 40.058  -20.432 39.853  /

\color{Red}
\put {\huge$\bullet$} at 10 37.3206

\color{black}
\normalcolor

\endpicture
\end{center}
\caption{Exact zeros (denoted by $\bullet$'s) converging to the point with argument equal to
$\frac{5}{12}\pi$ on the lima\c{c}on (determined by putting $\rho=\frac{1}{6}$ in equation \Ref{eqn33A}).
Approximate zeros (denoted by $\circ$'s) were computed from equation \Ref{eqn31}.  The exact zeros were computed
for $10\leq n \leq 150$ from $D_{2n}(a)$.}
\label{figure6}
\end{figure}

\section{Asymptotics for the leading zero}\label{sec:leadingzero}

Numerical exploration shows that the choice $k=1$ in equation \Ref{eqn31}
is a very good approximation to the leading zero $a^\ddagger_0$ 
(that is, of the zero of $D_{2n}(a)$ with smallest positive principal argument).  
The choice $k=2$ seems to approximate the next to leading zero as $n$ 
increases, but the choice $k=0$ is a spurious zero, which has no counterpart 
amongst the zeros of $D_{2n}(a)$. 

There are two short-comings about
equation \Ref{eqn31}.  The first is that the approximations leading to it introduced
a few new zeros, and the second is that, since all zeros approaches $2$ as $n\to\infty$
for fixed $k$, the approximations $a''_+$ must leave the set 
$S_\delta$ on which the approximation was done. This, for example, show that equation
\Ref{eqn31} is not an approximation to the $k$-th zero for fixed $k$, since
the $k$-th zero converges to $2$ and so leaves the set $S_\delta$.  On the
other hand, if $k=\lfloor \rho n\rfloor$, and $n\to\infty$, then equation \Ref{eqn31}
is an approximation for the zeros.  This is, for example, shown in figure \ref{figure6}
for the choice $\rho=\sfrac{1}{6}$, which corresponds to the point with argument
$\sfrac{5}{12}\pi$ on the lima\c{c}on.  Exact zeros and approximate zeros computed
from equation \Ref{eqn31} both converge to the limiting point.

In order to find asymptotic approximations of the location of the leading zero, 
we will approximate the partition function $D_{2n}(a)$ (see equation \Ref{eqn:exact_pf}).
The leading zero approaches $2$ at a rate proportional to $1/\sqrt{n}$, so, to first
order, $a_1 = 2 + O(1/\sqrt{n})$.   That is, we will consider 
$D_{2n}(2+\sfrac{c}{\sqrt{n}})$ and determine $c$.  Notice that
\begin{equation}
D_{2n}(2 + \sfrac{c}{\sqrt{n}})
= \sum_{k=0}^n\! \sum_{\ell=k}^n
\frac{2\ell+1}{n+\ell+1} \Bi{2n}{n+\ell} \Bi{\ell}{k} \frac{c^k}{n^{k/2}} .
\label{eqn33}   
\end{equation}
For fixed $n$ the summand is maximized when
\begin{eqnarray}
k &\simeq \frac{c^2 n}{(\sqrt{n}+c)(2\sqrt{n}+c)} 
= \sfrac{1}{2}c^2  - \sfrac{3}{4}\Sfrac{c^3}{\sqrt{n}}+ O\L \sfrac{1}{n} \R ; \\
\ell &\simeq \frac{ cn}{2\sqrt{n}+c}
= \sfrac{1}{2}c\sqrt{n} - \sfrac{1}{4}c^2 + O\L \Sfrac{1}{\sqrt{n}} \R .
\end{eqnarray}
Putting $n=m^2$ and $\ell = \lambda m$ and using Stirling's approximation to
approximate factorials in the binomial coefficients in the summand in equation
\Ref{eqn33} gives
\begin{equation}
\frac{2\ell+1}{n+\ell+1} \Bi{2n}{n+\ell} \Bi{\ell}{k} \frac{c^k}{n^{k/2}}
\simeq \frac{2^{2m^2+1}e^{-\lambda^2} c^k \lambda^{k+1}}{\sqrt{\pi}\, m^2 k!}
\, (1+o(1))
\label{eqn36}   
\end{equation}
to leading order.  The summation over $\ell$ can be approximated by integrating
this over $\lambda$ (using $\frac{d\ell}{d\lambda} = m$):
\begin{eqnarray}
\sum_{\ell=k}^n
\frac{2\ell+1}{n+\ell+1} \Bi{2n}{n+\ell} \Bi{\ell}{k} \frac{c^k}{n^{k/2}}
& \sim &
\int_0^\infty \frac{2^{2m^2+2}e^{-\lambda^2} c^k \lambda^{k+1}}{\sqrt{\pi}\, m^2 k!}m\,d\lambda \\
&=& \frac{2^{2m^2-k}\, c^k}{m\, \Gamma(\sfrac{1}{2} (k+1))}\, .
\end{eqnarray}
Summing over $k$ then gives the asymptotic formula
\begin{equation}
D_{2n}(2+\sfrac{c}{\sqrt{n}}) \sim
\Sfrac{4^n}{2\sqrt{\pi n}} \L
2+c\sqrt{\pi}\, e^{c^2/4} (1+\hbox{erf}(\sfrac{c}{2})) \R .
\end{equation}
The \textit{error function} is defined by
\begin{equation}
\hbox{erf} (x) = \Sfrac{2}{\sqrt{\pi}} \int_0^x e^{-t^2}\, dt
\end{equation}
for real $x$, and can be analytically continued to the entire complex plane. 
If $c=0$ then $D_{2n}(2) \sim \Sfrac{1}{\sqrt{\pi n}}\, 4^n$, which is the correct
asymptotics at $a=2$.

Zeros of $D_{2n}(a)$ are found at solutions of
\begin{equation}
F(c) = 2+c\sqrt{\pi}\, e^{c^2/4} (1+\hbox{erf}(\sfrac{c}{2})) = 0.
\label{eqn40}   
\end{equation}
The leading zero will correspond to that solution $c$ with smallest principal argument.
Solving numerically gives
\begin{equation}
c = 2.450314191845586\ldots+5.094256056412729\ldots\times \Imi 
\label{eqn41}   
\end{equation}
and this shows that the leading root approaches the critical point $a=2$ along
\begin{equation}
a_1 = 2 + \Sfrac{2.450314191845586\ldots}{\sqrt{n}}+
\Sfrac{5.094256056412729\ldots\times\Imi}{\sqrt{n}} + O\L \sfrac{1}{n} \R .
\label{eqn42}   
\end{equation}
Asymptotics for the next to leading zero can be determined by finding the
appropiate solution of equation \Ref{eqn40}. The result is
\begin{equation}
a_2 = 2 + \Sfrac{4.051192261300444\ldots}{\sqrt{n}}+
\Sfrac{6.323878106240248\ldots\times \Imi}{\sqrt{n}} + O\L \sfrac{1}{n} \R .
\label{eqn43}   
\end{equation}
The coefficients were easily computed to high precision and these results suggest that
these are not rational numbers.
The next term in equations \Ref{eqn42} and \Ref{eqn43} can be determined by 
computing the next term in equation \Ref{eqn36}.  This is
\begin{equation}
4^{m^2}\L \frac{2\, e^{-\lambda^2} c^k \lambda^{k+1}}{\sqrt{\pi}\, m^2 k!}
+ \frac{e^{-\lambda^2} (1+k+k^2-2\lambda^2)c^k\lambda^k}{
\sqrt{\pi}\, m^3 k!} + O\L \Sfrac{1}{m^4} \R \R .
\label{eqn44}   
\end{equation}
Computing an improved asymptotic formula for the partition from this gives
\begin{eqnarray}
\hspace{-2cm}
D_{2n}(2+\sfrac{c}{\sqrt{n}}) & \sim
\Sfrac{4^n}{2\sqrt{\pi n}} \L
2+c\sqrt{\pi}\, e^{c^2/4} (1+\hbox{erf}(\sfrac{c}{2})) \right. \nonumber \\
& \hspace{1cm} \left. - \Sfrac{c}{4\sqrt{n}} 
\L 2(2+c^2) + \sqrt{\pi} \, c (4+c^2)(1+\hbox{erf} \L \sfrac{c}{2} \R ) e^{c^2/4}  \R \R .
\label{eqn45}   
\end{eqnarray}
Put $c = c_1 + c_2 \sfrac{1}{\sqrt{n}}$ in the above, and set $c_1$ to be the value in equation \Ref{eqn41}, so as to eliminate the $\sfrac{1}{\sqrt{n}}$ term in \Ref{eqn45}.
The $\sfrac{1}{n}$ term in equation \Ref{eqn45} is then eliminated by putting
\begin{equation}
c_2 = -9.97370256476894\ldots + 12.482527911923\ldots \times \Imi .
\end{equation}
This improves the asymptotic for the leading zero $a_1$ in equation \Ref{eqn42} to
\begin{equation}
a_1 = 2 + \Sfrac{c_1}{\sqrt{n}} + \Sfrac{c_2}{n} + O \L \Sfrac{1}{\sqrt{n^3}} \R .
\label{eqn48}   
\end{equation}
In figure \ref{roots_approaching_2} the exact leading zero for $n\leq 150$ is shown by bullets
while the open circles are computed from equation \Ref{eqn48}.  The  approximation
is poor for small values of $n$, but improves when $n$ increases and the zeros
approaches the point $a=2$.

\begin{figure}[t]
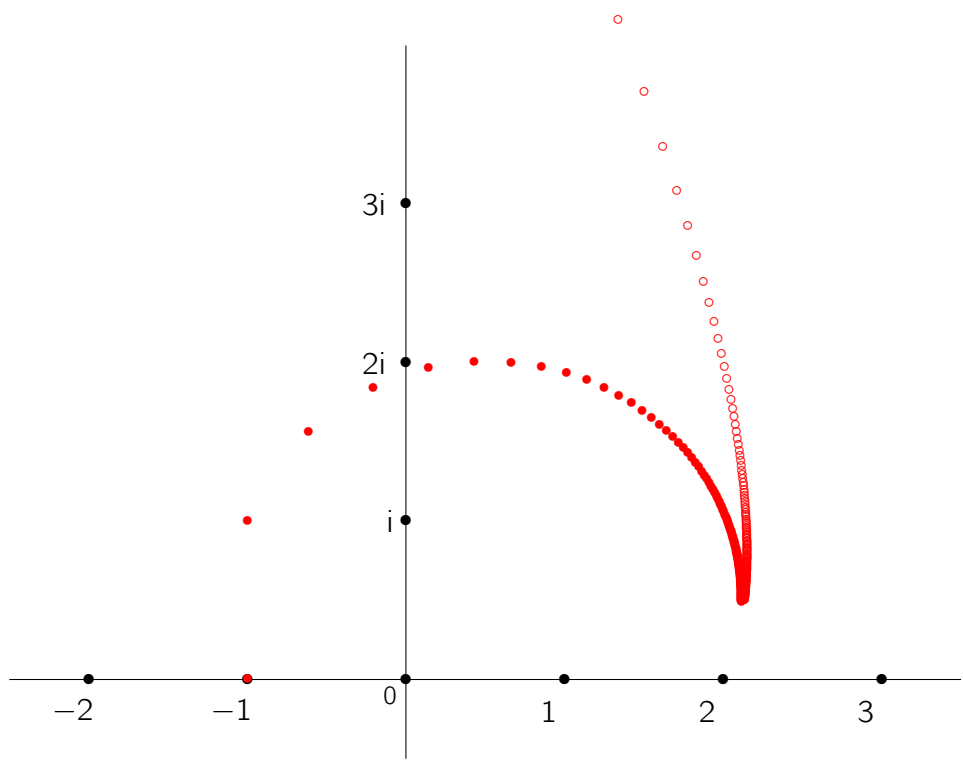

\begin{center}
\beginpicture
\color{black}
\setcoordinatesystem units <6pt,6pt>
\setplotarea x from -25 to 35, y from -5 to 40
\axes{-25}{-5}{0}{0}{35}{40}{-30}
\put {\footnotesize$0$} at -1 -1
\put {\footnotesize$\bullet$} at -20 0
\put {$-2$} at -21 -2
\put {\footnotesize$\bullet$} at -10 0
\put {$-1$} at -11 -2
\put {\footnotesize$\bullet$} at 10 0
\put {$1$} at 9 -2
\put {\footnotesize$\bullet$} at 20 0
\put {$2$} at 19 -2
\put {\footnotesize$\bullet$} at 30 0
\put {$3$} at 29 -2
\put {\footnotesize$\bullet$} at 0 10
\put {$\Imi$} at -1 10
\put {\footnotesize$\bullet$} at 0 20
\put {$2\Imi$} at -2 20
\put {\footnotesize$\bullet$} at 0 30
\put {$3\Imi$} at -2 30

\color{Red}
\multiput {\scriptsize$\bullet$} at -10.000 0.000  -10.000 10.000  
-6.145 15.639  -2.077 18.412  1.429 19.636  4.304 20.038  6.640 19.998  8.546 19.722  10.112 19.323  11.412 18.864  12.502 18.381  13.424 17.895  14.211 17.417  14.887 16.954  15.473 16.511  15.985 16.087  16.433 15.684  16.829 15.301  17.181 14.938  17.494 14.593  17.774 14.266  18.026 13.956  18.253 13.661  18.458 13.380  18.645 13.114  18.815 12.860  18.970 12.618  19.112 12.386  19.242 12.166  19.362 11.955  19.473 11.753  19.575 11.559  19.669 11.374  19.757 11.196  19.838 11.025  19.913 10.860  19.984 10.702  20.049 10.550  20.111 10.403  20.168 10.262  20.222 10.125  20.272 9.993  20.319 9.866  20.364 9.743  20.405 9.623  20.445 9.508  20.481 9.396  20.516 9.288  20.549 9.183  20.580 9.081  20.609 8.982  20.637 8.886  20.663 8.792  20.688 8.701  20.712 8.613  20.734 8.527  20.755 8.443  20.775 8.362  20.794 8.283  20.812 8.205  20.829 8.130  20.846 8.056  20.861 7.984  20.876 7.914  20.890 7.846  20.903 7.779  20.916 7.713  20.928 7.649  20.940 7.587  20.951 7.526  20.961 7.466  20.971 7.407  20.981 7.350  20.990 7.294  20.998 7.239  21.007 7.185  21.015 7.133  21.022 7.081  21.029 7.030  21.036 6.981  21.043 6.932  21.049 6.884  21.055 6.837  21.060 6.791  21.066 6.746  21.071 6.702  21.076 6.658  21.081 6.615  21.085 6.573  21.089 6.532  21.093 6.491  21.097 6.452  21.101 6.412  21.104 6.374  21.108 6.336  21.111 6.298  21.114 6.262  21.117 6.225  21.120 6.190  21.122 6.155  21.125 6.120  21.127 6.086  21.129 6.053  21.131 6.020  21.133 5.988  21.135 5.956  21.137 5.924  21.139 5.893  21.140 5.863  21.142 5.833  21.143 5.803  21.144 5.774  21.146 5.745  21.147 5.716  21.148 5.688  21.149 5.660  21.150 5.633  21.151 5.606  21.151 5.579  21.152 5.553  21.153 5.527  21.153 5.502  21.154 5.476  21.154 5.451  21.155 5.427  21.155 5.402  21.156 5.378  21.156 5.355  21.156 5.331  21.156 5.308  21.156 5.285  21.157 5.263  21.157 5.240  21.157 5.218  21.157 5.197  21.157 5.175  21.157 5.154  21.156 5.133  21.156 5.112  21.156 5.091  21.156 5.071  21.156 5.051  21.155 5.031  21.155 5.011  21.155 4.992  21.154 4.972  21.154 4.953  21.154 4.935  21.153 4.916   /

\color{Red}
\multiput {\scriptsize$\circ$} at 
13.380 
41.601  15.013 37.087  16.196 33.614  17.086 30.850  17.775 28.592  18.321 26.707  18.762 25.108  19.124 23.731  19.425 22.531  19.678 21.475  19.892 20.537  20.076 19.698  20.234 18.942  20.372 18.257  20.492 17.632  20.598 17.061  20.691 16.535  20.773 16.049  20.846 15.600  20.911 15.182  20.969 14.792  21.022 14.427  21.069 14.085  21.111 13.764  21.149 13.462  21.184 13.176  21.215 12.906  21.243 12.651  21.269 12.408  21.292 12.177  21.313 11.958  21.333 11.749  21.350 11.549  21.366 11.358  21.381 11.175  21.394 11.000  21.406 10.833  21.417 10.672  21.427 10.517  21.436 10.368  21.445 10.225  21.452 10.087  21.459 9.953  21.465 9.825  21.471 9.701  21.475 9.581  21.480 9.465  21.484 9.353  21.487 9.244  21.491 9.139  21.493 9.037  21.496 8.937  21.498 8.841  21.500 8.748  21.501 8.657  21.502 8.569  21.503 8.483  21.504 8.400  21.504 8.318  21.505 8.239  21.505 8.162  21.505 8.087  21.505 8.013  21.504 7.942  21.504 7.872  21.503 7.804  21.502 7.737  21.502 7.672  21.501 7.609  21.500 7.547  21.498 7.486  21.497 7.427  21.496 7.368  21.494 7.312  21.493 7.256  21.491 7.201  21.490 7.148  21.488 7.096  21.486 7.044  21.484 6.994  21.483 6.945  21.481 6.896  21.479 6.849  21.477 6.802  21.475 6.757  21.473 6.712  21.471 6.668  21.468 6.625  21.466 6.582  21.464 6.541  21.462 6.500  21.460 6.459  21.457 6.420  21.455 6.381  21.453 6.343  21.451 6.305  21.448 6.268  21.446 6.231  21.444 6.196  21.441 6.160  21.439 6.126  21.437 6.091  21.434 6.058  21.432 6.025  21.430 5.992  21.427 5.960  21.425 5.928  21.422 5.897  21.420 5.866  21.418 5.836  21.415 5.806  21.413 5.777  21.410 5.747  21.408 5.719  21.406 5.691  21.403 5.663  21.401 5.635  21.399 5.608  21.396 5.581  21.394 5.555  21.391 5.529  21.389 5.503  21.387 5.478  21.384 5.453  21.382 5.428  21.379 5.404  21.377 5.380  21.375 5.356  21.372 5.332  21.370 5.309  21.368 5.286  21.365 5.263  21.363 5.241  21.361 5.219  21.358 5.197  21.356 5.175  21.354 5.154  21.352 5.133  21.349 5.112  21.347 5.091  21.345 5.071  21.342 5.051  21.340 5.031  21.338 5.011  21.336 4.992   /


\color{black}
\normalcolor

\endpicture
\end{center}
\caption{Leading zeros (denoted by $\bullet$'s) and the asymptotic approximations
 to leading zeros (denoted by $\circ$'s).  The approximation to leading zeros
is given by $a_1$ in equation \Ref{eqn48}.  The approximation is poor for small
values of $n$, but inproves quickly with increasing values of $n$.
The exact zeros are given for $1\leq n \leq 150$ and the approximate zeros for
$6\leq n \leq 150$.}
\label{roots_approaching_2}
\end{figure}

\section{Conclusions}

In this paper we have examined the zeros of the partition function $D_{2n}(a)$ of 
adsorbing Dyck paths of length $2n$. These zeros are distributed over a region of 
the complex plane, but as $n$ becomes large they all (except for a single root at 
$a=0$) collect on a certain closed curve. We have shown that this curve is one 
lobe of a \lmcn{}, and that as $n\to\infty$ the roots become dense on this curve
(see theorem \ref{theorem4}).

In addition, we determined a formula approximating the locations of zeros for
finite values of $n$ (equation \Ref{eqn31}), and developed an asymptotic
formula for the location of the leading zeros as $n\to\infty$ (see equation
\Ref{eqn48}). In the limit as $n\to\infty$ the leading zeros converge to the point 
$a_c = 2$, forming an edge-singularity on the positive real axis.  The
rate of convergence to the edge-singularity is given in equation
\Ref{eqn48} as $O(1/\sqrt{n})$, and this is consistent with the crossover
exponent of adsorbing Dyck paths having value $\phi=\sfrac{1}{2}$ (see,
for example, references \cite{JvR10A,JvR15}). Recent work
suggests that the crossover exponent for adsorbing walks in three
dimensions may be different from $\sfrac{1}{2}$ \cite{BOP17}.

Adsorbing Dyck paths are just one example of a solvable model of 
interacting polymers. Another closely related model is that of 
\emph{pulled ballot paths}, which can be used to represent a linear 
polymer chain being pulled from a surface by an external force. 
In that case the curve on which the zeros accumulate turns out to 
be the outer boundary of two circles of radius $\sqrt{2}$, centred 
at $\pm \Imi$. There are many other solvable models in statistical 
mechanics, and the zeros of other models may display a variety of behaviours.

\vspace{1cm}

\noindent{\bf Acknowledgements:} EJJvR acknowledges financial support 
from NSERC (Canada) in the form of a Discovery Grant. NRB is supported 
by the Australian Research Council grant DE170100186.

\vspace{1cm}
\noindent{\bf References}
\bibliographystyle{plain}
\bibliography{PZerosDyck.bib}

\end{document}